\documentclass{article}
\usepackage{kr}
\usepackage{times}
\usepackage{soul}
\usepackage{url}
\usepackage{enumitem}
\usepackage[hidelinks]{hyperref}
\usepackage[utf8]{inputenc}
\usepackage[small]{caption}
\usepackage{graphicx}
\usepackage{mathtools}
\usepackage{amsmath}
\usepackage{amssymb}
\usepackage{amsthm}
\usepackage{booktabs}
\usepackage{algorithm}
\usepackage{algorithmic}
\usepackage{amsfonts}
\usepackage{xcolor}
\usepackage{multicol}
\usepackage{multirow}
\newcommand{\BD}{\mathsf{BD}}
\newcommand{\BDcirc}{\mathsf{BD}_\circ}
\newcommand{\BDtriangle}{\mathsf{BD}_\triangle}
\newcommand{\CPL}{\mathsf{CPL}}

\newcommand{\conp}{\mathsf{coNP}}
\newcommand{\np}{\mathsf{NP}}
\newcommand{\DP}{\mathsf{DP}}

\newcommand{\Prop}{\mathsf{Prop}}
\newcommand{\Lit}{\mathsf{Lit}}
\newcommand{\circLit}{\mathsf{Lit}_\circ}
\newcommand{\triangleLit}{\mathsf{Lit}_\triangle}
\newcommand{\Lcirctriangle}{\mathcal{L}_{\circ,\triangle}}
\newcommand{\Lcirc}{\mathcal{L}_\circ}
\newcommand{\Ltriangle}{\mathcal{L}_\triangle}
\newcommand{\LBD}{\mathcal{L}_\BD}

\newcommand{\true}{\mathbf{T}}
\newcommand{\both}{\mathbf{B}}
\newcommand{\neither}{\mathbf{N}}
\newcommand{\false}{\mathbf{F}}

\newcommand{\DNF}{\mathsf{DNF}}
\newcommand{\CNF}{\mathsf{CNF}}
\newcommand{\NNF}{\mathsf{NNF}}

\newcommand{\Hmsf}{{\mathsf{H}}}

\newcommand{\Pmsf}{{\mathsf{P}}}

\newcommand{\Xmbf}{\mathbf{X}}
\newcommand{\Ymbf}{\mathbf{Y}}

\newcommand{\Vmc}{{\mathcal{V}}}

\newcommand{\cl}{\mathsf{cl}}
\newcommand{\four}{\mathbf{4}}
\newcommand{\consvDashBD}{\models^{\mathsf{cons}}_\BD}
\newcommand{\consvDashCPL}{\models^{\mathsf{cons}}_\CPL}
\newcommand{\PropSol}{\mathcal{S}^\mathsf{p}}
\newcommand{\Sol}{\mathcal{S}}
\newcommand{\BDminSol}{\Sol^\BD}
\newcommand{\ThminSol}{\Sol^\mathsf{Th}}

\newtheorem{theorem}{Theorem}
\newtheorem{proposition}{Proposition}
\newtheorem{lemma}{Lemma}
\newtheorem{example}{Example}
\newtheorem{definition}{Definition}
\newtheorem{convention}{Convention}

\title{Abductive Reasoning in a~Paraconsistent Framework}
\author{%
Meghyn Bienvenu$^{1,2}$\and
Katsumi Inoue$^3$\and
Daniil Kozhemiachenko$^1$ \\
\affiliations
$^1$Universit\'{e} de Bordeaux, CNRS, Bordeaux INP, LaBRI, UMR 5800\\
$^2$ Japanese-French Laboratory for Informatics, CNRS, NII, IRL 2537, Tokyo, Japan\\ 
$^3$National Institute of Informatics, Tokyo, Japan\\
\emails
\{meghyn.bienvenu,daniil.kozhemiachenko\}@u-bordeaux.fr,
inoue@nii.ac.jp
}
\begin{document}
\allowdisplaybreaks
\maketitle
\begin{abstract}
We explore the problem of explaining observations starting from a classically inconsistent theory by adopting a paraconsistent framework. We consider two expansions of the well-known Belnap--Dunn paraconsistent four-valued logic $\BD$: $\BDcirc$ introduces formulas of the form $\circ\phi$ (‘the information about $\phi$ is reliable’), while $\BDtriangle$ augments the language with formulas $\triangle\phi$ (‘there is information that $\phi$ is true’). We define and motivate the notions of abduction problems and explanations in $\BDcirc$ and $\BDtriangle$ and show that they are not reducible to one another. We analyse the complexity of standard abductive reasoning tasks (solution recognition, solution existence, and relevance / necessity of hypotheses) in both logics. Finally, we show how to reduce abduction in $\BDcirc$ and $\BDtriangle$ to abduction in classical propositional logic, thereby enabling the reuse of existing abductive reasoning procedures. 
\end{abstract}
\section{Introduction\label{sec:introduction}}
Logic-based abduction is an important form of reasoning with multiple applications in artificial intelligence, including diagnosis and commonsense reasoning \cite{EiterGottlob1995}. An \emph{abduction problem} can be generally formulated as a pair $\langle\Gamma,\chi\rangle$ consisting of a set of formulas $\Gamma$ (\emph{theory}) and a~formula $\chi$ (\emph{observation}) s.t.\ $\Gamma\not\models\chi$, and the task is to find an \emph{explanation}, i.e., a formula $\phi$ s.t.\ $\Gamma,\phi\models\chi$. Of course, not \emph{every} formula is intuitively acceptable as an explanation which is why there are usually some restrictions on $\phi$. In particular, $\Gamma\cup\{\phi\}$ should be consistent, $\phi$~should not entail $\chi$ by itself nor contain atoms not occurring in $\Gamma\cup\{\phi\}$, $\phi$ should be syntactically restricted so as to be easily understandable, and $\phi$ should constitute a weakest possible (or minimal) explanation, cf.\ discussion in ~\cite[\S4.2]{Marquis2000HDRUMS} or~\cite[\S3.3]{Aliseda2006abductivereasoningbook}. Most commonly, the third desiderata is enforced by requiring abductive solutions to take the form of \emph{terms} (conjunctions of literals), in which case the logically weakest solutions are simply the subset-minimal ones. 

Note, however, that in classical propositional logic ($\CPL$) any \emph{contradictory} theory is inconsistent. Thus, there is no explanation of $\chi$ from a~contradictory~$\Gamma$. This can be circumvented in two ways. First, by \emph{repairing} $\Gamma$, i.e., making it consistent and then proceeding as usual (cf.,~e.g.,~\cite{DuWangShen2015}). Second, by moving to a~\emph{paraconsistent} logic. The characteristic feature of such logics is the failure of the explosion principle --- $p,\neg p\not\models q$.
\paragraph{Abduction in Paraconsistent Logics}
The question of how to employ paraconsistent logics to perform abductive reasoning on classically inconsistent theories has already generated interest in the philosophical logic community. For example, \cite{Carnielli2006} considers abduction in a~three-valued logic called \emph{Logic of Formal Inconsistency} ($\mathbf{LFI1}$), obtaining by expanding the language of $\CPL$ (classical propositional logic) with new connectives $\bullet\phi$ and $\circ\phi$ (read ‘$\phi$ has a~non-classical value’ or ‘the information about $\phi$ is unreliable’ and ‘$\phi$ has a~classical value’ or ‘the information about $\phi$ is reliable’, respectively). More recently, \cite{Bueno-SolerCarnielliConiglioRodriguesFilho2017} and~\cite{ChlebowskiGajdaUrbanski2022} considered abductive explanations in the \emph{minimal Logic of Formal Inconsistency} ($\mathbf{mbC}$), and \cite{RodriguesConiglioAntunesBueno-SolerCarnielli2023} consider abduction in a~four-valued \emph{Logic of Evidence and Truth} ($\mathsf{LET_K}$).

These studies showcase the interest of paraconsistent abduction, but as they issue from a different research community, the formulation of abductive solutions and the questions that are explored depart from those typically considered in knowledge representation and reasoning (KR). In particular, these works allow \emph{arbitrary} formulas as solutions (rather than terms). Moreover, to the best of our knowledge, there are no results on the complexity of paraconsistent reasoning tasks (e.g., solution existence). 
Also, in $\mathbf{mbC}$ and $\mathsf{LET_K}$, $\circ$ and $\bullet$ do not have truth-functional semantics, which complicates the comparison with classical logic and the reuse of established techniques. 

\paragraph{Abduction in Belnap--Dunn Logic}
The preceding considerations motivate us to revisit paraconsistent abduction by taking a KR perspective and adopting the~well-known paraconsistent propositional logic $\BD$ by~\cite{Dunn1976,Belnap1977fourvalued,Belnap1977computer}. The main idea of $\BD$ is to treat the values of formulas as the information an agent (or a~computer as in~\cite{Belnap1977computer}) might have w.r.t.\ a~given statement $\phi$. This results in four ‘Belnapian’ values:
\begin{itemize}
\item $\true$ --- ‘the agent is only told that $\phi$ is true’;
\item $\false$ --- ‘the agent is only told that $\phi$ is false’;
\item $\both$ --- ‘the agent is told that $\phi$ is false and that it is true’;
\item $\neither$ --- ‘the agent is not told that $\phi$ is false nor that it is true’.
\end{itemize}
The truth and falsity conditions of $\neg$, $\wedge$, and $\vee$ are defined in a~classical manner but assumed to be \emph{independent}.
\begin{center}
\begin{tabular}{c|c|c}
&\textbf{is true when}&\textbf{is false when}\\\hline
$\neg\phi$&$\phi$ is false&$\phi$ is true\\
$\phi_1\wedge\phi_2$&$\phi_1$ and $\phi_2$ are true&$\phi_1$ or $\phi_2$ is false\\
$\phi_1\vee\phi_2$&$\phi_1$ or $\phi_2$ is true&$\phi_1$ and $\phi_2$ are false
\end{tabular}
\end{center}
One can see from the table above that there are no \emph{valid formulas} (that always have value in $\{\true,\both\}$, i.e., \emph{at least true}) over $\{\neg,\wedge,\vee\}$. Likewise, there are no formulas that are always \emph{at least non-true} (have value in $\{\neither,\false\}$). Thus, defining $\phi\models_\BD\chi$ as ‘in every valuation $v$ s.t.\ $v(\phi)\in\{\true,\both\}$, $v(\chi)\in\{\true,\both\}$ as well’, we obtain that $p\wedge\neg p\not\models_\BD q$.

Note also that, as already observed in \cite{RodriguesConiglioAntunesBueno-SolerCarnielli2023}, the $\{\neg,\wedge,\vee\}$-language is too weak for abduction: there is no solution for $\langle\{p\vee q\},q\rangle$ except for $q$ itself because $p\vee q,\neg p\not\models_\BD q$. However, if one assumes that $p$ has value $\false$, this will explain $q$ provided $p\vee q$. To do this, one needs to expand the language of $\BD$ with new connectives.

One option is to formulate abductive solutions in $\BDcirc$, the expansion of $\BD$ with the \emph{truth-functional} version of $\circ$ by~\cite{OmoriWaragai2011,OmoriSano2014TiL}. The connective $\circ$ is interpreted as follows: $\circ\phi$ has value $\true$ if $\phi$ has value $\true$ or $\false$, and has value $\false$, otherwise. In this expanded language, the formula $\neg p\wedge\circ p$, which expresses that 
\emph{there is reliable information that $p$ is false}, yields a solution to $\langle\{p\vee q\},q\rangle$. 
Another possibility is to adopt $\BDtriangle$, the expansion of $\BD$ with $\triangle$ from~\cite{SanoOmori2014}. Here, $\triangle\phi$ can be interpreted as ‘there is information that $\phi$ is true’ and has the following semantics: $v(\triangle\phi)=\true$ if $v(\phi)\in\{\true,\both\}$ and $v(\triangle\phi)=\false$, otherwise. In this case, the formula $\neg p\wedge\neg\triangle p$ (which reads: \emph{$p$ is false and there is no information that it is true}) is a~solution.

Importantly, $\BDtriangle$ and $\BDcirc$ will allow us to solve some abduction problems that do not admit any solutions 
in classical logic. 
The following example adapted from~\cite[Example~5]{RodriguesConiglioAntunesBueno-SolerCarnielli2023} shows how one can deal with explanations from \emph{classically inconsistent theories}.
\begin{example}\label{example:impossible1}
Assume that a~valuable item was stolen and we have contradictory information about the theft: it was stolen by either Paula or Quinn, but both of them claim to have an alibi. This can be represented with $\Gamma=\{p\vee q,\neg p,\neg q\}$. Classically, there is no explanation for $p$ nor $q$ w.r.t.\ the theory $\Gamma$. However, in $\BDcirc$, we can explain $q$ by assuming that \emph{Paula's alibi was confirmed by a~reliable source}, which can be represented by the formula $\circ p$. This abduction problem can also be solved in $\BDtriangle$ as follows: $q$ can be explained by assuming that \emph{Paula's alibi was not disputed} --- $\neg\triangle p$.
\end{example}
\paragraph{Contributions}
In this paper, we formalize abductive reasoning in $\BD$ and explore its computational properties. We consider abductive solutions defined as terms in $\BDtriangle$ and $\BDcirc$ and compare the classes of abduction problems that can be solved in $\BDtriangle$ and $\BDcirc$. We establish an (almost complete) picture of the complexity of standard abductive reasoning tasks (solution recognition, solution existence, and relevance and necessity of hypotheses) in $\BDtriangle$ and $\BDcirc$, for two entailment-based notions of minimality. 
We show how $\BD$ abduction problems can be reduced to abduction in classical propositional logic (and vice-versa), which makes it possible 
to apply classical consequence-finding methods to generate abductive solutions 
in $\BDtriangle$ and $\BDcirc$.
\paragraph{Plan of the Paper}
The remainder of the text is organised as follows. In Section~\ref{sec:BD}, we formally introduce $\BD$, $\BDcirc$, and $\BDtriangle$ and discuss their semantical and computational properties. In Section~\ref{sec:abduction}, we present the notions of abduction problem and explanation in expansions of $\BD$. Sections~\ref{sec:complexity} and~\ref{sec:solvingBDproblems} are dedicated to the complexity of solving $\BD$ abduction problems. Section~\ref{sec:problemembeddings} discusses embeddings of $\BD$ abduction problems into the classical framework. Finally, in Section~\ref{sec:conclusion}, we summarise the paper's results and outline a~plan for future work. Due to limited space, some proofs have been put in the appendix of the extended version~\cite{longversion}.
\section{$\BD$ and its Expansions\label{sec:BD}}
Since $\BD$, $\BDcirc$, and $\BDtriangle$ use the same set of truth values and the same $\neg$, $\wedge$, and $\vee$, we present their syntax and semantics together 
in the following definition.
\begin{definition}\label{def:BDsemantics}
The language $\Lcirctriangle$ is constructed from a~fixed countable set of propositional variables $\Prop$ via the following grammar:
\begin{align*}
\Lcirctriangle\ni\phi&\coloneqq p\in\Prop\mid\neg\phi\mid(\phi\wedge\phi)\mid(\phi\vee\phi)\mid\circ\phi\mid\triangle\phi
\end{align*}
In what follows, we use $\LBD$, $\Lcirc$ and $\Ltriangle$ to denote the fragments of $\Lcirctriangle$ over $\{\neg,\wedge,\vee\}$, $\{\neg,\wedge,\vee,\circ\}$, and $\{\neg,\wedge,\vee,\triangle\}$, respectively. We will also use $\bullet\phi$ as a~shorthand for ${\neg\circ}\phi$ and use $\Prop(\phi)$ to denote the set of all variables occurring in a formula~$\phi$.

We set $\four=\{\true,\both,\neither,\false\}$ and define a \emph{$\BD$ valuation} as a~mapping $v:\Prop\rightarrow\four$ that is extended to complex formulas as follows.
\begin{center}
\normalsize{
\begin{tabular}{ccc}
\begin{tabular}{c|cccc}
$\wedge$ & $\true$ & $\both$ & $\neither$ & $\false$ \\\hline
$\true$ & $\true$ & $\both$ & $\neither$ & $\false$ \\
$\both$ & $\both$ & $\both$ & $\false$ & $\false$ \\
$\neither$ & $\neither$ & $\false$ & $\neither$ & $\false$ \\
$\false$ & $\false$ & $\false$ & $\false$ & $\false$
\end{tabular}
&
\begin{tabular}{c|c}
&$\neg$\\\hline
$\true$&$\false$\\
$\both$&$\both$\\
$\neither$&$\neither$\\
$\false$&$\true$ 
\end{tabular}
&
\begin{tabular}{c|c}
&$\circ$\\\hline
$\true$&$\true$\\
$\both$&$\false$\\
$\neither$&$\false$\\
$\false$&$\true$ 
\end{tabular}
\\
&\\
\begin{tabular}{c|cccc}
$\vee$ & $\true$ & $\both$ & $\neither$ & $\false$ \\\hline
$\true$ & $\true$ & $\true$ & $\true$ & $\true$ \\
$\both$ & $\true$ & $\both$ & $\true$ & $\both$ \\
$\neither$ & $\true$ & $\true$ & $\neither$ & $\neither$ \\
$\false$ & $\true$ & $\both$ & $\neither$ & $\false$
\end{tabular}
&
\begin{tabular}{c|c}
&$\triangle$\\\hline
$\true$&$\true$\\
$\both$&$\true$\\
$\neither$&$\false$\\
$\false$&$\false$ 
\end{tabular}
&
\begin{tabular}{c|c}
&$\bullet$\\\hline
$\true$&$\false$\\
$\both$&$\true$\\
$\neither$&$\true$\\
$\false$&$\false$ 
\end{tabular}
\end{tabular}}
\end{center}

A~formula $\phi$ is \emph{$\BD$-valid}, written $\BD\models\phi$, iff $\forall v:v(\phi)\in\{\true,\both\}$, and $\phi$ is \emph{$\BD$-satisfiable} iff $\exists v:v(\phi)\in\{\true,\both\}$. A~set of formulas $\Gamma$~\emph{entails}~$\chi$, written $\Gamma\models_\BD\chi$, iff
\begin{align*}
\forall v\left[(\forall\phi\in\Gamma~v(\phi)\in\{\true,\both\})\Rightarrow v(\chi)\in\{\true,\both\}\right]
\end{align*}
We will henceforth use $\phi\simeq\chi$ ($\phi$ and $\chi$ are \emph{weakly equivalent}) to denote that $\phi$ and $\chi$ entail one another and $\phi\equiv\chi$ ($\phi$ and $\chi$ are \emph{strongly equivalent}) to denote that $\phi$ and $\chi$ have the same Belnapian value in every $\BD$ valuation.
\end{definition}

Let us briefly discuss the intuitive interpretation of $\circ$ and~$\bullet$. 
Since we do not model sources and multiple agents, ‘the information concerning $p$ is reliable’ ($\circ p$) can be construed as an assumption on the part of the reasoning agent that the available information concerning $p$ is trustworthy, meaning that if we have the information that $p$ holds, then we cannot have evidence for its negation, not $p$ (and likewise, if we have information that $p$ doesn't hold, there is no evidence for $p$ holding). The negation, $\bullet p$, which is termed ‘unreliable’ but may be more precisely phrased as ‘unavailable or unreliable’, then covers two cases: the absence of information to support or refute $p$ (so, neither $p$ nor its negation can be inferred) or the presence of contradictory information (both the statement and its negation follow).

It is important to note that there is no $\BD$-valid $\phi\in\LBD$ since $\neither$ is preserved by $\neg$, $\wedge$, and $\vee$. In $\BDcirc$ and $\BDtriangle$, however, we can define
\begin{align*}
\top_{\BDcirc}&\coloneqq{\circ\circ}p&\bot_{\BDcirc}&\coloneqq{\bullet\circ}p\\
\top_{\Ltriangle}&\coloneqq\triangle p\vee\neg\triangle p&\bot_{\Ltriangle}&=\triangle p\wedge\neg\triangle p
\end{align*}
and check that for every valuation $v$,
$$
v(\top_{\BDcirc})\!=\!v(\top_{\BDtriangle})\!=\!\true~\qquad v(\bot_{\BDcirc})\!=\!v(\bot_{\BDtriangle})\!=\!\false
$$

One can also notice that $\BDcirc$ is less expressive than $\BDtriangle$. Indeed, $\circ\phi\equiv(\triangle\phi\wedge\neg\triangle\neg\phi)\vee(\triangle\neg\phi\wedge\neg\triangle\phi)$, while on the other hand, $\triangle$ cannot be defined via $\circ$~\cite[Corollaries~6.1 and~6.24]{OmoriSano2015}. It is also easy to check that distributive and De Morgan laws hold w.r.t.\ $\neg$, $\wedge$, and $\vee$, and that the following equivalences hold for $\triangle$: 
\begin{align}\label{equ:triangleequivalence}
\triangle\triangle\phi&\equiv\triangle\phi&\neg\triangle\neg\triangle\phi&\equiv\triangle\phi\nonumber\\
\triangle(\phi\wedge\chi)&\equiv\triangle\phi\wedge\triangle\chi&\triangle(\phi\vee\chi)&\equiv\triangle\phi\vee\triangle\chi
\end{align}
Thus, every $\phi\in\Ltriangle$ can be transformed into a strongly equivalent formula $\NNF(\phi)$ in~\emph{negation normal form}, i.e., built from literals of the form $p$, $\neg p$, $\triangle p$, $\neg\triangle p$, $\triangle\neg p$, and $\neg\triangle\neg p$ using $\wedge$ and $\vee$.

One can also show that in $\BDcirc$ contraposition holds w.r.t.\ $\neg$ while in $\BDtriangle$ w.r.t.\ $\neg\triangle$ and that the deduction theorem can be recovered using $\triangle$.
\begin{proposition}\label{prop:contraposition}
Let $\phi,\chi\in\Lcirc$ and $\varrho,\sigma,\tau\in\Ltriangle$. Then the following statements hold.
\begin{enumerate}
\item $\phi\models_\BD\chi$ iff $\neg\chi\models_\BD\neg\phi$.
\item $\varrho,\sigma\models_\BD\tau$ iff $\varrho,\neg\triangle\tau\models_\BD\neg\triangle\sigma$.
\item $\varrho,\sigma\models_\BD\tau$ iff $\varrho\models_\BD\neg\triangle\sigma\vee\tau$.
\end{enumerate}
\end{proposition}

Note that Proposition~\ref{prop:contraposition} entails that $\phi\simeq\chi$ iff $\phi\equiv\chi$ in $\BDcirc$. On the other hand, this is not the case in $\BDtriangle$: $p\simeq\triangle p$ but $p$ is not strongly equivalent to $\triangle p$. Still, we can show that some $\triangle$'s can be removed from formulas while preserving weak equivalence.
\begin{definition}\label{def:flatformulas}
Let $\phi\in\Lcirctriangle$. We use $\phi^\flat$ to denote the result of removing $\triangle$'s that are \emph{not in the scope of $\neg$} from $\NNF(\phi)$.
\end{definition}
\begin{proposition}\label{prop:triangleequivalence}
Let $\phi\in\Lcirctriangle$. Then, $\phi\simeq\phi^\flat$. 
\end{proposition}

We finish the section by establishing faithful embeddings of $\CPL$ into $\BDcirc$, $\BDtriangle$, and $\BD$.
\begin{proposition}\label{prop:CPLtoBDtrianglecirc}
Let $\phi\in\LBD$. We let $\phi^\triangle$ denote the result of replacing each occurrence of each variable $p$ in $\phi$ with $\triangle p$ and $\phi^\circ$ denote the result of replacing $p$'s with $\circ p$'s. Then $\phi$ is $\CPL$-valid iff $\phi^\triangle$ is $\BD$-valid iff $\phi^\circ$ is $\BD$-valid.
\end{proposition}
Finally, $\CPL$ can be embedded even in $\BD$ itself.
\begin{proposition}\label{prop:CPLtoBD}
Let $\phi,\chi\in\LBD$ and $\Prop[\{\phi,\chi\}]=\{p_1,\ldots,p_n\}$. Then $\phi\models_\CPL\chi$ iff $$\phi\wedge\bigwedge\limits^{n}_{i=1}(p_i\vee\neg p_i)\models_\BD\chi\vee\bigvee\limits^{n}_{i=1}(p_i\wedge\neg p_i).$$
\end{proposition}

Notice that the embeddings of $\CPL$ into $\BDcirc$ and $\BDtriangle$ increase the size of $\phi$ only linearly. Thus, since $\CPL$-validity is $\conp$-complete, $\BD$-validity of $\Lcirc$ and $\Ltriangle$ formulas is also $\conp$-complete. Likewise, $\BD$ entailment is $\conp$-complete even for $\LBD$-formulas since the embedding of $\models_\CPL$ into $\models_\BD$ increases the size of formulas only linearly.\footnote{Membership in $\conp$ is immediate since $\BD$ and the considered expansions have truth-table semantics.}
\section{Abduction in $\BDtriangle$ and $\BDcirc$\label{sec:abduction}}
We begin the presentation of abduction in $\BDtriangle$ and $\BDcirc$ with definitions of literals, terms, and clauses in $\Lcirc$ and $\Ltriangle$. Note that since $\circ$ and $\bullet$ do not distribute over $\wedge$ and $\vee$, we cannot assume that in $\Lcirc$, literals do not contain binary connectives, terms are $\vee$-free and clauses are $\wedge$-free \emph{if we want every formula to be representable as a~conjunction of clauses or a~disjunction of terms}.
\begin{definition}[Literals, terms, and clauses]\label{def:literalstermsclauses}~
\begin{itemize}
\item \emph{Propositional literal} is a variable $p$ or its negation $\neg p$.
\item \emph{$\Ltriangle$-literal} has one of the following forms: $p$, $\neg p$, $\triangle p$, $\neg\triangle p$, $\triangle\neg p$, $\neg\triangle\neg p$ ($p\in\Prop$). \emph{$\Ltriangle$-clause} is a~disjunction of literals; \emph{$\Ltriangle$-term} is a~conjunction of literals.
\item \emph{$\Lcirc$-literal} has one of the following forms: $p$, $\neg p$, $\circ\phi$, $\bullet\phi$ ($\phi\in\Lcirc$). Clauses and terms are defined as above.
\end{itemize}
\end{definition}
The next statement is immediate.
\begin{proposition}\label{prop:DNFCNF}
Let $\phi\in\Lcirc$ and $\chi\in\Ltriangle$. Then there are conjunctions of $\Lcirc$- and $\Ltriangle$-clauses $\DNF(\phi)$ and $\DNF(\chi)$, and disjunctions of $\Lcirc$- and $\Ltriangle$-terms $\CNF(\phi)$ and $\CNF(\chi)$ s.t.\ $\phi\equiv\DNF(\phi)\equiv\CNF(\phi)$ and $\chi\equiv\DNF(\chi)\equiv\CNF(\chi)$.
\end{proposition}

Note, however, that even though the definition of clauses and terms we gave above allows for a~strongly equivalent representation of every $\Lcirc$-formula, the clauses and terms may be difficult to interpret in natural language. Indeed, while $\circ p$ ($\bullet p$) can be understood as ‘information concerning $p$ is (un)reliable’, a~formula such as $\bullet(p\wedge\circ(\neg q\vee\bullet(r\wedge\neg s)))$ does not have any obvious~ natural-language interpretation. It thus makes sense to consider \emph{atomic} $\Lcirc$-literals.
\begin{definition}\label{def:atomicliterals}
\emph{Atomic $\Lcirc$-literals} are formulas of a~form $p$, $\neg p$, $\circ p$\footnote{Note that since $\circ\neg\phi\equiv\circ\phi$, we do not need to consider literals $\circ\neg p$ and $\bullet\neg p$.}, or $\bullet p$ with $p\in\Prop$. \emph{Atomic $\Lcirc$-clauses} are disjunctions of atomic literals and \emph{atomic $\Lcirc$-terms} are conjunctions of atomic literals.
\end{definition}
\begin{convention}
For a~propositional literal $l$, we set $\overline{p}=\neg p$, $\overline{\neg p}=p$ and $\overline{\overline{l}}=l$. Given a~formula $\phi$, we use $\Lit(\phi)$, $\circLit(\phi)$, and $\triangleLit(\phi)$ to denote the sets of propositional literals, atomic $\Lcirc$-literals, and $\Ltriangle$-literals occurring in $\phi$.
\end{convention}

We can now define abduction problems and solutions. 
\begin{definition}[Abduction problems and solutions]\label{def:abductiveproblem}~
\begin{itemize}
\item A \emph{$\BD$ abduction problem} is a~tuple $\mathbb{P}=\langle\Gamma,\psi,\Hmsf\rangle$ with $\Gamma\cup\{\psi\}\subseteq\LBD$, $\Gamma\not\models_\BD\psi$, and $\Hmsf$ is a~finite set of $\Ltriangle$-literals or atomic $\Lcirc$-literals. If $\Hmsf$ is not restricted, we will omit it for brevity. We call $\Gamma$~a~\emph{theory}, members of $\Hmsf$ \emph{hypotheses}, and $\chi$ an \emph{observation}.
\item An~\emph{$\Lcirc$-solution} of $\mathbb{P}$ is an atomic $\Lcirc$-term $\tau$ composed from the literals in $\Hmsf$ s.t.
$\Gamma,\tau\models_\BD\psi$ and
$\Gamma,\tau\not\models_\BD\bot$.
\item An~\emph{$\Ltriangle$-solution} is an $\Ltriangle$-term composed from the literals in $\Hmsf$ s.t.
$\Gamma,\tau\models_\BD\psi$ and
$\Gamma,\tau\not\models_\BD\bot$. 
\item A~solution $\tau$ is \emph{proper} if $\tau\not\models_\BD\psi$. 
\item A proper solution $\tau$ is \emph{$\models_\BD$-minimal} if there is no proper solution $\phi$ s.t.\ $\phi\not\simeq\tau$, $\tau\models_\BD\phi$. 
\item A~proper solution $\tau$ is \emph{theory-minimal} if there is no proper solution $\phi$ s.t.\ $\Gamma,\phi\not\models_\BD\tau$ and $\Gamma,\tau\models_\BD\phi$.
\end{itemize}
\end{definition}
\begin{convention}
Given an abduction problem $\mathbb{P}$, we will use $\Sol(\mathbb{P})$, $\PropSol(\mathbb{P})$, $\BDminSol(\mathbb{P})$, and $\ThminSol(\mathbb{P})$ to denote the sets of all solutions, all proper solutions, all $\models_\BD$-minimal solutions, and all theory-minimal solutions of $\mathbb{P}$, respectively.
\end{convention}

\begin{convention}
To simplify the presentation of examples, we shall sometimes omit $\Hmsf$ when specifying abduction problems. In such cases, it is assumed that $\Hmsf$ contains all possible $\Ltriangle$-literals or atomic $\Lcirc$-literals (depending on which language we are considering).
\end{convention}

Observe that we have two notions of minimality. The first one ($\models_\BD$-minimality or, entailment-minimality) generalises \emph{subset-minimality} by~\cite{EiterGottlob1995} to the $\BD$ setting. As we will see in Section~\ref{ssec:entailmentbetweenterms}, even though entailment between $\Ltriangle$- and atomic $\Lcirc$-terms is polynomially decidable, it is not equivalent to the containment of one term in the other, whence we need a~more general criterion. This approach to minimality is also presented by~\cite{Aliseda2006abductivereasoningbook}. Theory-minimal solutions are, essentially, \emph{least specific} solutions in the terminology of~\cite{Stickel1990,SakamaInoue1995} or \emph{least presumptive} solutions in the terminology of~\cite{Poole1989}. Theory-minimal solutions can also be seen as duals of \emph{theory prime implicates} from~\cite{Marquis1995}.

In addition, it is easy to see that even though a~theory-minimal solution is $\models_\BD$-minimal, the converse need not hold. Indeed, consider $\mathbb{P}=\langle\{p\vee q,r\},q\wedge r\rangle$. In both $\Ltriangle$ and $\Lcirc$, there are two $\models_\BD$-minimal solutions: $\neg\triangle p$ and $q$ in $\Ltriangle$, and $\neg p\wedge\circ p$ and $q$ in $\Lcirc$. But $\neg\triangle p$ and $\neg p\wedge\circ p$ are \emph{not theory-minimal}: $p\vee q,r,\neg\triangle p\models_\BD q$ and $p\vee q,r,\neg p\wedge\circ p\models_\BD q$.

One can also notice that allowing \emph{any} $\Lcirc$-terms (rather than only atomic $\Lcirc$-terms) results in even weaker solutions. Consider for example the following problem $\mathbb{P}=\langle\Gamma,\chi,\Hmsf\rangle$ (for the sake of discussion, we locally abuse notation and terminology by admitting non-atomic $\Lcirc$-terms as solutions).
\begin{align}\label{equ:circleterm}
\Gamma&=\{(p\wedge p')\vee(q\wedge q'),\neg(p\wedge p')\}&\chi&=q\wedge q'\nonumber\\\Hmsf&=\{p,p',\circ p,\circ p',\circ(p\wedge p')\}
\end{align}
It is clear that $\circ(p\wedge p')$ and $\circ p\wedge\circ p'$ solve~\eqref{equ:circleterm} and that $\circ p\wedge\circ p'\models_\BD\circ(p\wedge p')$. Moreover, one can see that $\circ(p\wedge p')$ is a~theory-minimal solution. Furthermore, some abduction problems cannot be solved if we only allow atomic terms. Indeed, one can check that no atomic $\Lcirc$-term over $p$ and $q$ properly solves $\langle\{p\vee q\},(p\vee\neg p)\wedge(q\vee\neg q)\rangle$. On the other hand, $\circ p\wedge\circ((p\vee\neg p)\wedge(q\vee\neg q))$ is a~solution.

In what follows, we will illustrate the differences between abduction in $\BDcirc$, $\BDtriangle$, and classical logic. 
%
First, we can observe that some problems have abductive solutions both in $\Ltriangle$ and $\Lcirc$ (cf.~Example~\ref{example:impossible1}). On the other hand, some problems can be solved only in $\Lcirc$. That is, there are no solutions in the form of $\Ltriangle$-terms\footnote{Of course, there are \emph{$\Ltriangle$-formulas} that can solve the problem but we are interested in solutions in the form of terms.} even though there are $\Lcirc$-terms that solve the problem.
\begin{example}\label{example:impossible2}
As in Example~\ref{example:impossible1}, either Paula or Quinn is culpable, but now there is also evidence that implicates Paula --- $p$. In this case, we want to justify that \emph{Paula is innocent} --- $\neg p$. There is, of course, no proper classical solution for $\langle\{p\vee q,p\},\neg p\rangle$. Likewise, one can check that this problem admits no proper solutions in~$\Ltriangle$, as $\tau \models_\BD\neg p$ for any $\Ltriangle$-term $\tau$ s.t.\ $p\vee q,p,\tau\models_\BD \neg p$.

How can we solve the problem in $\BDcirc$? We can add the 
$\Lcirc$-term $\bullet p$, i.e., assume that the \emph{evidence against Paula is unreliable}. This way, we have $p\vee q,p,\bullet p\models_\BD\neg p$, which is justified since one must not be convicted on unreliable evidence. It can be verified that $\bullet p$ is the unique $\BD$-minimal proper solution for $\langle\{p\vee q,p\},\neg p\rangle$.
\end{example}

Conversely, some problems can be solved only in $\Ltriangle$.
\begin{example}\label{example:impossible3}
Consider $\mathbb{P}=\langle\{p\vee\neg p\vee q\},q\rangle$. To explain $q$, one must assume that $p$ has value~$\neither$. Formally, this means that we assume $\neg\triangle p\wedge\neg\triangle\neg p$. One can check that this is a~theory-minimal proper solution. On the other hand, it is easy to see that there is no atomic $\Lcirc$-term (and, in fact, no $\Lcirc$-formula at all) that can solve $\langle\{p\vee\neg p\vee q\},q\rangle$.
\end{example}

Since $\Lcirc$- and $\Ltriangle$-solutions to $\BD$ abduction problems are incomparable, it makes sense to ask which $\Lcirc$-solutions can be represented as $\Ltriangle$-solutions and vice versa. In the remainder of the section, we answer this question.
\begin{definition}\label{def:determinedNfreeterms}~
\begin{itemize}[noitemsep,topsep=1pt]
\item A~satisfiable $\Ltriangle$-term $\tau$ is \emph{$\neither$-free} if for every $\neg\triangle l$ occurring in $\tau$, $\triangle\overline{l}$ or $\overline{l}$ also occurs in $\tau$.
\item A~satisfiable atomic $\Lcirc$-term $\tau$ is \emph{determined} if for every $p$ s.t.\ $\circ p$ or $\bullet p$ occurs in $\tau$, $p$ or $\neg p$ also occurs in~$\tau$.
\item An $\Lcirc$-solution (resp.\ $\Ltriangle$-solution) $\varrho$ of a~$\BD$ abduction problem $\mathbb{P}$ is $\Ltriangle$-representable (resp.\ $\Lcirc$-representable) if there exists an $\Ltriangle$-solution (resp.\ $\Lcirc$-solution) $\sigma$ of $\mathbb{P}$ s.t.\ $\varrho\simeq\sigma$.
\end{itemize}
\end{definition}
\begin{theorem}\label{theorem:representablesolutions}~
\begin{enumerate}[noitemsep,topsep=1pt]
\item An $\Lcirc$-solution is $\Ltriangle$-representable iff it is determined.
\item An $\Ltriangle$-solution is $\Lcirc$-representable iff it is $\neither$-free.
\end{enumerate}
\end{theorem}
\begin{proof}
For Statement 1, let $\tau$ be an $\Lcirc$-solution of some abduction problem. By definition, $\tau$ is satisfiable. Suppose that $\tau$ is determined. 
As $p\wedge\neg p\equiv\neg p\wedge\bullet p\equiv p\wedge\bullet p\equiv p\wedge\neg p\wedge\bullet p$, we can assume w.l.o.g.\ that $\tau$ has the following form:
\begin{align*}
\tau&=\bigwedge^{m}_{i=1}(l_i\wedge\circ l_i)\wedge\bigwedge^{m'}_{i'=1}(l'_{i'}\wedge\bullet l'_{i'})\wedge\bigwedge\limits^{n}_{j=1}l''_j
\end{align*}
where the $l_i$, $l'_{i'}$, and $l''_j$ are propositional literals. 
Now, we observe that $l\wedge\circ l\equiv\triangle l\wedge\neg\triangle\overline{l}$. Thus, $\tau$ can be represented by the following $\neither$-free term $\tau^{\circ\triangle}$:
\begin{align*}
\tau^{\circ\triangle}&=\bigwedge^{m}_{i=1}(\triangle l_i\wedge\neg\triangle\overline{l_i})\wedge\bigwedge^{m'}_{i'=1}(l'_{i'}\wedge\overline{l'_{i'}})\wedge\bigwedge\limits^{n}_{j=1}l''_j
\end{align*}
For the converse, suppose $\tau$ is not determined, i.e., there is some $\circ p$ s.t.\ neither $p$ nor $\neg p$ occur in $\tau$ or there is some $\bullet q$ s.t.\ neither $q$ not $\neg q$ occur in $\tau$. By examining the truth table semantics from Definition~\ref{def:BDsemantics}, we can see that there is no conjunction of $\Ltriangle$-literals that is weakly equivalent to $\circ p$ or $\bullet q$. Hence, $\tau$ is not $\Ltriangle$-representable.

For Statement 2, consider an $\Ltriangle$-solution $\tau$. First, suppose that $\tau$ is an $\neither$-free term, and let $\tau^\flat$ be as in Definition~\ref{def:flatformulas}. We have that $\tau^\flat\simeq\tau$ with $\tau^\flat$ having the following form:
\begin{align*}
\tau^\flat&=\bigwedge\limits^{m}_{i=1}(l_i\wedge\neg\triangle\overline{l}_i)\wedge\bigwedge\limits^{n}_{j=1}l'_j
\end{align*}
It is clear that the following atomic $\Lcirc$-term represents $\tau^\flat$ (and hence, $\tau$):
\begin{align*}
\tau^{\triangle\circ}&=\bigwedge\limits^{m}_{i=1}(l_i\wedge\circ l_i)\wedge\bigwedge\limits^{n}_{j=1}l'_j
\end{align*}
For the converse, suppose $\tau$ is \emph{not $\neither$-free}, and w.l.o.g.\ let $\neg\triangle p$ occur in $\tau^\flat$ but $p$ not occur. From~\cite[Propositions~6.23 and~6.23]{OmoriSano2015}, we know that $\neg\triangle p$ (by itself, without $p$) is not definable in $\BDcirc$. As $\tau$ is satisfiable (being a solution), this means $\tau$ is not $\Lcirc$-representable. 
\end{proof}
We finish the section with a few observations. First, one can see that determined and $\neither$-free terms can be recognised in polynomial time. Second, we note that none of the solutions in Examples~\ref{example:impossible1}--\ref{example:impossible3} is representable in the other language. In Example~\ref{example:impossible1}, $\neg\triangle p$ is not $\Lcirc$-representable and $\circ p$ is not $\Ltriangle$-representable. In Example~\ref{example:impossible2}, $\bullet p$ is not $\Ltriangle$-representable. In Example~\ref{example:impossible3}, $\neg\triangle p\wedge\neg\triangle\neg p$ is not $\Lcirc$-representable. This shows that even though $\BDcirc$ is less expressive than $\BDtriangle$, their sets of solutions are incomparable. Finally, we remark that even if a~problem has solutions in both languages, the solutions need not 
be (weakly or strongly) equivalent, especially if we consider $\BD$- or theory-minimal solutions. Indeed, one can see that the solutions in Example~\ref{example:impossible1} are theory-minimal, yet they are not equivalent. More than that, neither 
 implies the other.
\section{Complexity of Term Entailment\label{sec:complexity}}
This section contains some technical results concerning entailment from $\Ltriangle$- and $\Lcirc$-terms that facilitate the proofs of complexity bounds in Section~\ref{sec:solvingBDproblems}. 
\subsection{Entailment Between Terms\label{ssec:entailmentbetweenterms}}
We begin with the complexity of entailment between terms. Recall from Definition~\ref{def:abductiveproblem} that to establish the $\models_\BD$-minimality of $\tau$, we need to check whether there is another solution $\phi$ that is entailed by~$\tau$. In the next two theorems, we show that the entailment of atomic $\Lcirc$-terms and $\Ltriangle$-terms is recognisable in polynomial time.
\begin{theorem}\label{theorem:atomiccirctermPtime}
Entailment between atomic $\Lcirc$-terms is decidable in deterministic polynomial time.
\end{theorem}
\begin{proof}[Proof sketch]
Let $\sigma$ and $\sigma'$ be atomic $\Lcirc$-terms. We begin by noting that $\sigma$ is $\BD$-unsa\-tis\-fi\-able iff (i) $p$, $\neg p$, and $\circ p$ occur in $\sigma$; or (ii) $\circ p$ and $\bullet p$ occur in $\sigma$. The ‘only if’ direction is evident since $p\wedge\neg p\wedge\circ p$ and $\circ p\wedge\bullet p$ are unsatisfiable. For the ‘if’ direction, assume that there is no variable $p$ s.t.\ $p$, $\neg p$, and $\circ p$ occur in $\sigma$, nor any variable $q$ s.t.\ $\circ q$ and $\bullet q$ occur in $\sigma$. We construct a~satisfying valuation $v$ as follows:
\begin{itemize}
\item $v(r)=\true$ iff $r$ occurs in $\sigma$ but $\neg r$ and $\bullet r$ \emph{do not};
\item $v(r)=\false$ iff $\neg r$ occurs in $\sigma$ but $r$ and $\bullet r$ do not;
\item $v(r)=\both$ otherwise.
\end{itemize}
It is clear that $v(\sigma)\in\{\true,\both\}$. Indeed, it is easy to check that every \emph{literal} occurring in $\sigma$ has value $\true$ or $\both$ and that $v$ is well defined. 

It follows from this characterisation that the satisfiability of an atomic $\Lcirc$-term can be decided in polynomial time. In addition, observe that $\sigma\models_\BD\sigma'$ if $\sigma$ is unsatisfiable and $\sigma\not\models_\BD\sigma'$ is $\sigma$ is satisfiable but $\sigma'$ is not. Finally, to show that entailment between \emph{satisfiable} atomic $\Lcirc$-terms is decidable in polynomial time, we use the facts that $p\wedge\neg p\models_\BD\bullet p$ and $p\wedge\neg p\equiv p\wedge\bullet p\equiv\neg p\wedge\bullet p$.
\end{proof}

To establish a~similar property of $\Ltriangle$-terms, we construct a~faithful embedding of $\Ltriangle$-formulas into the language of $\CPL$. The polynomial complexity evaluation of term entailment will follow since our embedding increases the size of formulas only \emph{linearly} and since determining entailment of terms in $\CPL$ can also be done in polynomial time.
\begin{definition}\label{def:BDtriangletoCPL}
Let $\phi\in\Ltriangle$ be in $\NNF$ and let ${\sim}$\footnote{While we may assume that $\CPL$ is given over $\LBD$, we sometimes use $\sim$ rather than $\neg$ to make clear we are working in $\CPL$.} denote the \emph{classical negation}. We define $\phi^\cl$ as follows.
\begin{align*}
p^\cl&=p^+&(\neg p)^\cl&=p^-\\
(\triangle p)^\cl&=p^+&(\triangle\neg p)^\cl&=p^-\\
(\neg\triangle p)^\cl&={\sim}p^+&(\neg\triangle\neg p)^\cl&={\sim}p^-\\
(\chi\wedge\psi)^\cl&=\chi^\cl\wedge\psi^\cl&(\chi\vee\psi)^\cl&=\chi^\cl\vee\psi^\cl
\end{align*}
\end{definition}
\begin{lemma}\label{lemma:BDtriangletoCPL}
Let $\phi,\!\chi\!\in\!\Ltriangle$ be in $\NNF$. Then $\phi\!\models_\BD\!\chi$ iff $\phi^\cl\!\models_\CPL\!\chi^\cl$.
\end{lemma}
The next theorem follows immediately from Lemma \ref{lemma:BDtriangletoCPL}.
\begin{theorem}\label{theorem:triangletermsPtime}
Entailment between $\Ltriangle$-terms is decidable in deterministic polynomial time.
\end{theorem}

To check that $\tau$ is a~theory-minimal solution to $\langle\Gamma,\chi\rangle$, we need to establish that there is no other solution $\sigma$ s.t.\ $\Gamma,\tau\models_\BD\sigma$ but $\Gamma,\sigma\not \models_\BD\tau$. The next results show that the complexity of checking term entailment w.r.t.\ a background theory differs depending on whether we consider $\Lcirc$ or $\Ltriangle$.
\begin{theorem}\label{theorem:weakminimalitycircconphard}
It is $\conp$-complete to decide whether $\Gamma,\varrho\models_\BD\sigma$, given $\Gamma\subseteq\LBD$ and atomic $\Lcirc$-terms $\varrho$ and~$\sigma$. 
\end{theorem}
\begin{proof}
Membership 
is evident since $\BDcirc$ has truth-table semantics. For hardness, observe 
that $\Gamma$ is \emph{classically unsatisfiable} iff $\Gamma,\bigwedge\limits_{p\in\Prop[\Gamma]}\circ p\models_\BD q$ with $q\notin\Prop[\Gamma]$.
\end{proof}

In $\Ltriangle$, however, this problem is tractable. 
\begin{theorem}\label{theorem:weakminimalitytriangleptime}
It can be decided in polynomial time whether $\Gamma,\varrho\models_\BD\sigma$ , given  $\Gamma\subseteq\LBD$ and $\Ltriangle$-terms $\varrho$ and $\sigma$.
\end{theorem}
\begin{proof}
It suffices to show the result for the case where $\sigma$ is an $\Ltriangle$-literal. Due to Propositions~\ref{prop:contraposition} and \ref{prop:triangleequivalence}, $\Gamma,\varrho\models_\BD\sigma$ iff $\Gamma,\varrho,(\neg\triangle\sigma)^\flat\models_\BD\bot$, so we may focus on solving the latter task. First, we assume w.l.o.g.\ that all formulas in $\Gamma$ are in $\NNF$ and that $\varrho\wedge(\neg\triangle\sigma)^\flat$ is $\BD$-satisfiable (this can be checked in polynomial time by Theorem~\ref{theorem:triangletermsPtime}). By Lemma~\ref{lemma:BDtriangletoCPL}, we have that $\Gamma,\varrho,(\neg\triangle\sigma)^\flat\models_\BD\bot$ iff $\Gamma^\cl,\varrho^\cl,((\neg\triangle\sigma)^\flat)^\cl\models_\CPL\bot$. Since $\Gamma\subseteq\LBD$, $\Gamma^\cl$ is $\sim$-free (cf.~De\-fi\-ni\-tion~\ref{def:BDtriangletoCPL}). As $\varrho\wedge(\neg\triangle\sigma)^\flat$ is $\BD$-satisfiable by assumption, then so is $\varrho^\cl\wedge((\neg\triangle\sigma)^\flat)^\cl$ (by Lemma~\ref{lemma:BDtriangletoCPL}). Thus, there is no variable $r$ s.t.\ $r$ and ${\sim}r$ both occur in $\varrho^\cl\wedge((\neg\triangle\sigma)^\flat)^\cl$. Now take $\Gamma^\cl$ and substitute every $p$ that occurs in $\varrho^\cl\wedge((\neg\triangle\sigma)^\flat)^\cl$ positively (resp.\ negatively) with $\top$ (resp.\ $\bot$). We then exhaustively apply the following $\CPL$-equivalence-preserving transformations to all subformulas of~$\bigwedge\limits_{\phi\in\Gamma^\cl}\phi$, denoting the result by $(\Gamma^\cl)^\sharp$:
\begin{align}
\top\wedge\psi&\rightsquigarrow\psi&\top\vee\psi&\rightsquigarrow\top&\bot\wedge\psi&\rightsquigarrow\bot&\bot\vee\psi&\rightsquigarrow\psi\label{equ:01reduction}
\end{align}
Clearly, $(\Gamma^\cl)^\sharp$ can be computed in polynomial time in 
the size of~$\Gamma^\cl$. Moreover, if $(\Gamma^\cl)^\sharp\!=\!\bot$, then $\Gamma^\cl,\varrho^\cl,((\neg\triangle\sigma)^\flat)^\cl\models_\CPL\bot$ holds. To complete the proof, we show that $\Gamma^\cl,\varrho^\cl,((\neg\triangle\sigma)^\flat)^\cl \not \models_\CPL\bot$ when $(\Gamma^\cl)^\sharp\!\neq\!\bot$. Let us assume $(\Gamma^\cl)^\sharp\!\neq\!\bot$ and set $v(r)\!=\!\true$ for every $r\!\in\!\Prop((\Gamma^\cl)^\sharp)$. It is clear that $v((\Gamma^\cl)^\sharp)=\true$ since $(\Gamma^\cl)^\sharp$ is $\sim$-free (cf.~Definition~\ref{def:BDtriangletoCPL}). Now, we extend $v$ to the variables occurring in $\varrho^\cl$ and $((\neg\triangle\sigma)^\flat)^\cl$ as expected: if $s$ occurs, we set $v(s)=\true$, if $v({\sim}s)$ occurs, we set $v(s)=\false$. As $\varrho^\cl\wedge((\neg\triangle\sigma)^\flat)^\cl$ is $\CPL$-satisfiable, we now have $v\left(\bigwedge\limits_{\phi\in\Gamma^\cl}\phi\wedge\varrho^\cl\wedge((\neg\triangle\sigma)^\flat)^\cl\right)=\true$, as required.
\end{proof}
\subsection{Entailment of Formulas From Terms}
\begin{table}[t]
\centering
\begin{tabular}{lcc}
Decision problem &$\Ltriangle$&$\Lcirc$\\\midrule
$\tau\models_\BD\psi$?&$\text{ in }\mathsf{P}$&$\conp$\\[.2em]
$\tau\models_\BD\sigma$?&$\text{ in }\mathsf{P}$&$\text{ in }\mathsf{P}$\\[.2em]
$\Gamma,\tau\models_\BD\sigma$?&$\text{ in }\mathsf{P}$&$\conp$\\[.2em]
$\Gamma,\tau\not\models_\BD\bot$?&$\text{ in }\mathsf{P}$&$\np$\\[.2em]
$\Gamma,\tau\models_\BD\psi$?&$\conp$&$\conp$
\end{tabular}
\caption{Complexity of entailment tasks, where $\Gamma\cup\{\psi\}\subseteq\LBD$, and $\sigma$ and $\tau$ are terms in the considered language ($\Ltriangle$ or $\Lcirc$). Unless specified otherwise, all results are completeness results.}
\label{tab:complexityparts}
\end{table}
Let us now consider the complexity of entailment of \emph{$\LBD$-formulas}\footnote{Recall that in $\BD$ abduction problems, the theory is formulated in $\LBD$ while solutions are formulated in $\Lcirc$ or~$\Ltriangle$.} by $\Lcirc$-terms and $\Ltriangle$-terms. We begin with $\Lcirc$-terms. The following statement is straightforward.
\begin{theorem}\label{theorem:circtermsformulasconp}
It is is $\conp$-complete to decide 
whether $\phi\models_\BD\chi$, given an atomic $\Lcirc$-term $\phi$ and $\chi\in\LBD$. 
\end{theorem}
\begin{proof}
Membership 
is immediate as $\BDcirc$-entailment is in $\conp$. To show  $\conp$-hardness, observe that $\chi$ is $\CPL$-valid iff $\bigwedge\limits_{p\in\Prop(\chi)}\!\!\!\!\circ p\models_\BD\chi$ since $\circ p$ ensures $v(p)\in\{\true,\false\}$ and $\wedge$, $\vee$, and $\neg$ behave classically on $\true$ and $\false$.
\end{proof}

On the other hand, if $\phi$ is an $\Ltriangle$-term, $\phi\models_\BD\chi$ can be decided in deterministic polynomial time.
\begin{theorem}\label{theorem:triangletermsformulasPtime}
It can be decided in polynomial time whether $\phi\models_\BD\chi$, given an $\Ltriangle$-term $\phi$ and $\chi\in\LBD$.
\end{theorem}
\begin{proof}
Note first, that if $\Prop(\phi)\cap\Prop(\chi)=\varnothing$ and $\phi$ is satisfiable, then $\phi\not\models_\BD\chi$. Indeed, we can just evaluate all variables of $\chi$ as $\neither$ which will make $\chi$ have value $\neither$ as well (cf.~Definition~\ref{def:BDsemantics}). By Theorem~\ref{theorem:triangletermsPtime}, it takes polynomial time to determine whether an $\Ltriangle$-term is satisfiable. Thus, we consider the case when $\phi$ is satisfiable.

We assume w.l.o.g.\ that $\chi$ is in $\NNF$. By Lemma~\ref{lemma:BDtriangletoCPL}, we have that $\phi\models_\BD\chi$ iff $\phi^\cl\models_\CPL\chi^\cl$. Since $\chi\in\LBD$, $\chi^\cl$ is $\sim$-free (cf.~Definition~\ref{def:BDtriangletoCPL}). The rest of the proof is similar to that of Theorem~\ref{theorem:weakminimalitytriangleptime}. The only difference is that we will check whether $\chi^\cl$ is reduced to $\top$ using~\eqref{equ:01reduction}.
\end{proof}

Finally, we observe that in the presence of $\Gamma$, the entailment of $\LBD$-formulas from terms becomes $\conp$-complete.
\begin{theorem}\label{theorem:BD+termconpcomplete}
It is $\conp$-complete to decide whether $\Gamma,\tau\models_\BD\psi$, given $\Gamma\cup\{\psi\}\in\LBD$ and an $\Ltriangle$-term or an atomic $\Lcirc$-term $\tau$. 
\end{theorem}
\begin{proof}
The proof is a straightforward reduction from the $\conp$-complete entailment problem for $\LBD$-formulas: $\phi \models_\BD\chi$ iff $\phi\vee p,\neg p \wedge \circ p\models_\BD\chi$ iff $\phi \vee p,\neg p \wedge \neg\triangle p \models_\BD \chi$. To see why the preceding holds, note that $(\phi\vee p)\wedge(\neg p\wedge\circ p)\equiv\phi$ and $\neg p\wedge\neg\triangle p\equiv\neg p\wedge\circ p$.
\end{proof}

We summarise the results of the section in Table~\ref{tab:complexityparts}. Note that the polynomial decidability of $\Gamma,\tau\not\models_\BD\bot$ for $\Ltriangle$-terms follows immediately from Theorem~\ref{theorem:weakminimalitytriangleptime}. Similarly, $\np$-completeness of $\Gamma,\tau\not\models_\BD\bot$ for atomic $\Lcirc$-terms is a~corollary of Theorem~\ref{theorem:BD+termconpcomplete}. To see it, use a~fresh $p$ for $\psi$ and set $\tau=\bigwedge\limits_{q\in\Prop[\Gamma]}\circ q$. It is immediate that $\Gamma,\bigwedge\limits_{q\in\Prop[\Gamma]}\circ q\models_\BD p$ iff $\Gamma$ is classically unsatisfiable.
\section{Complexity of $\BD$ Abduction\label{sec:solvingBDproblems}}
This section considers the complexity of the principal decision problems related to $\BD$ abduction, namely, solution recognition, solution existence, and relevance and necessity of hypotheses. Table \ref{tab:complexitysolutions} summarises the obtained results for $\BD$ abduction, alongside results for classical abduction.\footnote{Results for $\CPL$ come from~\cite{EiterGottlob1995} or can be obtained as corollaries of the latter paper or our own results.}

We begin with a~useful technical lemma.
\begin{definition}\label{def:consistententailment}
We say that $\Gamma$ \emph{$\BD$-consistently entails} $\chi$ ($\Gamma\consvDashBD\chi$) iff $\Gamma$ is $\BD$-satisfiable and $\Gamma\models_\BD\chi$.
\end{definition}
\begin{lemma}\label{lemma:BDconsistententailmentcomplexity}
Let $\Gamma\cup\{\chi\}\subseteq\LBD$, $\sigma$ be an $\Ltriangle$-term, and $\tau$ be an atomic $\Lcirc$-term. Then
\begin{enumerate}[noitemsep,topsep=1pt]
\item deciding whether $\Gamma,\sigma\consvDashBD\chi$ is $\conp$-complete;
\item deciding whether $\Gamma,\tau\consvDashBD\chi$ is $\DP$-complete.
\end{enumerate}
\end{lemma}
\begin{proof}
We begin with Statement 1. $\conp$-membership is immediate since entailment is in $\conp$ and verifying the consistency of $\Gamma,\sigma$ can be done in polynomial time using Theorem~\ref{theorem:weakminimalitytriangleptime}. For hardness, we reduce $\BD$-entailment to $\BD$-consistent entailment as follows:
\begin{align*}
\phi\models_\BD\chi&\text{ iff }\phi\vee p,q\consvDashBD\chi\vee p\tag{$p,q\notin\Prop[\{\phi,\chi\}]$, $\sigma=q$}
\end{align*}
Let $\phi\not\models_\BD\chi$. We can thus find $
v$ such that $v(\phi)\in\{\true,\both\}$ and $v(\chi)\notin\{\true,\both\}$. Since $p$ and $q$ are fresh, we let $v(p)=\false$ and $v(q)=\true$ which falsifies the consistent entailment. Conversely, let $\phi\vee p,q\not\consvDashBD\chi\vee p$. It is clear that $\{\phi\vee p,q\}$ is $\BD$-satisfiable. Thus, there is a~valuation s.t.\ $v(\phi\vee p)\in\{\true,\both\}$, $v(q)\in\{\true,\both\}$ but $v(\chi\vee p)\notin\{\true,\both\}$. Hence, $v(\phi)\in\{\true,\both\}$ and $v(\chi)\notin\{\true,\both\}$, as required.

For Statement 2, membership follows immediately from Table~\ref{tab:complexityparts}. To show hardness, we reduce the $\DP$-complete $\mathsf{Sat}$-$\mathsf{UnSat}$ problem for $\CPL$ to  $\BD$-consistent entailment. Let $\phi,\chi$ be propositional formulas 
and assume w.l.o.g.\ that $\Prop(\phi)\cap\Prop(\chi)=\varnothing$. Set $\Xi=\Prop(\phi)\cup\Prop(\chi)$ and pick $p\notin\Xi$. We show that $\phi$ is $\CPL$-satisfiable and $\chi$ is $\CPL$-unsatisfiable iff $\phi,p\wedge\bigwedge\limits_{q\in\Xi}\circ q\consvDashBD p\wedge\neg\chi$.

If $\phi$ is $\CPL$-unsatisfiable, then $\phi\wedge p\wedge\bigwedge\limits_{q\in\Xi}\circ q$ is $\BD$-unsatisfiable, so the consistent entailment fails. If $\chi$ is $\CPL$-satisfiable, let $v$ be a~classical valuation s.t.\ $v(\chi)=\true$ (whence, $v(p\wedge\neg\chi)=\false$) and $v(\phi\wedge p)=\true$ (recall that $\Prop(\phi)\cap\Prop(\chi)=\varnothing$, so such valuation must exist unless $\phi$ is $\CPL$-unsatisfiable). Again, the consistent entailment fails.

For the converse, let $\phi$ be $\CPL$-satisfiable and $\chi$ is $\CPL$-unsatisfiable. It is clear that $\phi,p\wedge\bigwedge\limits_{q\in\Xi}\circ q\consvDashBD p\wedge\neg\chi$ because $p\wedge\bigwedge\limits_{q\in\Xi}\circ q\models_\BD p\wedge\neg\chi$ and $\phi\wedge p\wedge\bigwedge\limits_{q\in\Xi}\circ q$ is $\BD$-satisfiable as $p\notin\Prop(\phi)$ and $\phi$ is $\CPL$-satisfiable.
\end{proof}
\subsection{Solution Recognition}
We use the preceding lemma to establish the complexity of recognising arbitrary and proper solutions. 
\begin{theorem}\label{theorem:anyproperBDtrianglesolution}
It is $\conp$-complete to decide, given a~$\BD$ abduction problem $\mathbb{P}$ and an $\Ltriangle$-term $\sigma$, 
whether $\sigma$ is a~(proper) solution of $\mathbb{P}$.
\end{theorem}
\begin{proof}
$\conp$-completeness of recognising $\sigma\in\Sol(\mathbb{P})$ follows immediately from Lemma~\ref{lemma:BDconsistententailmentcomplexity} since $\sigma$ is a~solution of $\langle\Gamma,\chi,\Hmsf\rangle$ iff $\Gamma,\chi\consvDashBD\chi$. For proper solutions, we observe that recognising an~\emph{arbitrary} solution is reducible to the recognition of~\emph{a~proper} solution as follows: if we let $p\notin\Prop[\Gamma\cup\{\chi\}\cup\Hmsf]$, then $\sigma$ is a~solution of $\langle\Gamma,\chi,\Hmsf\rangle$ iff $\sigma$ is a~\emph{proper solution} of $\left\langle\Gamma\cup\left\{p\right\},p\wedge\chi,\Hmsf\cup\{p\}\right\rangle$. Indeed, $\Gamma,\sigma\models_\BD\chi$ holds iff $\Gamma,p,\sigma\models_\BD\chi\wedge p$ holds, and likewise, $\Gamma,p,\sigma\models_\BD\bot$ iff $\Gamma,\sigma\models_\BD\bot$. Moreover, since $p$ does not occur in $\tau$, it is clear that $\sigma\not\models_\BD\chi\wedge p$. This establishes $\conp$-hardness. For membership, we note that checking the properness condition ($\sigma \not \models_\BD \chi$) is in $\mathsf{P}$ since $\sigma$ is an $\Ltriangle$-term (Theorem \ref{theorem:triangletermsformulasPtime}), so we remain in $\conp$. 
\end{proof}
\begin{theorem}\label{theorem:anyproperBDcircsolution}
It is $\DP$-complete to decide, given a~$\BD$ abduction problem $\mathbb{P}$ and an atomic $\Lcirc$-term $\sigma$, 
whether $\sigma$ is a~(proper) solution of $\mathbb{P}$.
\end{theorem}
\begin{proof}
The arguments are the same as in the proof of Theorem~\ref{theorem:anyproperBDtrianglesolution}, 
but use the complexity results for atomic $\Lcirc$-terms from Table~\ref{tab:complexityparts} 
and Lemma \ref{lemma:BDconsistententailmentcomplexity}. 
\end{proof}
\begin{table}
\centering
\begin{tabular}{lccc}
Recognition and existence & $\Ltriangle$ & $\Lcirc$& $\CPL$ \\ \midrule 
$\tau\in\Sol(\mathbb{P})$? / $\tau\in\PropSol(\mathbb{P})$? & $\conp$ & $\DP$& $\DP$ \\[.2em]
$\tau\in\BDminSol(\mathbb{P})$? & $\DP$ & $\DP$& $\DP$ \\[.2em]
$\tau\in\ThminSol(\mathbb{P})$? & in $\Pi^\Pmsf_2$ & in $\Pi^\Pmsf_2$& in $\Pi^\Pmsf_2$\\[.2em]
$\Sol(\mathbb{P})\neq\varnothing$? / $\PropSol(\mathbb{P})\neq\varnothing$? & $\Sigma^\Pmsf_2$ & $\Sigma^\Pmsf_2$& $\Sigma^\Pmsf_2$\\[.75em]
Relevance& $\Ltriangle$ & $\Lcirc$& $\CPL$ \\\midrule
w.r.t.\ $\Sol(\mathbb{P})$, $\PropSol(\mathbb{P})$, $\BDminSol(\mathbb{P})$ & $\Sigma^\Pmsf_2$ & $\Sigma^\Pmsf_2$& $\Sigma^\Pmsf_2$\\[.2em]
w.r.t.\ $\ThminSol(\mathbb{P})$ & in $\Sigma^\Pmsf_3$ & in $\Sigma^\Pmsf_3$& in $\Sigma^\Pmsf_3$\\[.75em]
Necessity& $\Ltriangle$ & $\Lcirc$& $\CPL$ \\\midrule
w.r.t.\ $\Sol(\mathbb{P})$, $\PropSol(\mathbb{P})$, $\BDminSol(\mathbb{P})$ & $\Pi^\Pmsf_2$ & $\Pi^\Pmsf_2$& $\Pi^\Pmsf_2$\\[.2em]
w.r.t.\ $\ThminSol(\mathbb{P})$ & in $\Pi^\Pmsf_3$ & in $\Pi^\Pmsf_3$& in $\Pi^\Pmsf_3$
\end{tabular}
\caption{Complexity of abductive reasoning problems. Unless specified otherwise, all results are completeness results.}
\label{tab:complexitysolutions}
\end{table}

We next show how to recognize $\models_\BD$-minimal solutions. 
\begin{theorem}\label{theorem:minimalrecognition}
It is $\DP$-complete to decide, given a~$\BD$ abduction problem $\mathbb{P}$ and a term $\sigma$, whether $\sigma$ is a~$\models_\BD$-minimal ($\Lcirc$- or $\Ltriangle$-) solution of $\mathbb{P}$.
\end{theorem}
\begin{proof}[Proof sketch]
The upper bound exploits the fact that for every term $\sigma$, we can identify polynomially many terms $\sigma_1, \ldots, \sigma_k$ s.t.\ (i) $\sigma \models_\BD \sigma_i$ but $\sigma_i \not \models_\BD \sigma$ ($1 \leq i \leq k$), and (ii) if $\sigma$ is a proper solution but not $\models_\BD$-minimal, then some $\sigma_i$ is a (proper) solution. Thus, to verify $\models_\BD$-minimality, it suffices to check that none of the $\sigma_i$ is a solution. 

For $\Ltriangle$-terms, we may assume (recall Proposition~\ref{prop:triangleequivalence}) that all $\triangle$'s in $\sigma$ occur under $\neg$, i.e., $\sigma=\sigma^\flat$. In this case, $\sigma^\flat\models_\BD\sigma'^\flat$ iff $\triangleLit(\sigma^\flat)\supseteq\triangleLit(\sigma'^\flat)$, so we need only to check each of the (linearly many) terms $\sigma^{-l}$ obtained by deleting one $\Ltriangle$-literal from $\sigma$.

For atomic $\Lcirc$-terms, however, we cannot just remove literals because one term may $\BD$-entail another even though they have no common atomic literals, as in $p\wedge\neg p\models_\BD\bullet p$. Nevertheless, we can still identify syntactically a polynomial number of candidate better terms, each obtained by picking a variable $p$ occurring in $\sigma$ and replacing the set of $p$-literals in $\sigma$ by a $\models_\BD$-weaker set of $p$-literals. 
\end{proof}

In the case of theory-minimal solutions, we establish a $\Pi^\Pmsf_2$ upper bound. We expect that this case is indeed harder than $\models_\BD$-minimality, intuitively because the presence of the theory means we cannot readily identify a polynomial number of candidate better solutions to check. We leave the search for a~matching lower bound for future work and remark that, to the best of our knowledge, the complexity of the analogous problem in $\CPL$ is also unknown. 
\begin{theorem}\label{theorem:theoryminimalrecognition}
It is in $\Pi^\Pmsf_2$ to decide, given a~$\BD$ abduction problem $\mathbb{P}$ and an (atomic $\Lcirc$- or $\Ltriangle$-)term  $\sigma$,
whether $\sigma$ is a theory-minimal solution of $\mathbb{P}$.
\end{theorem}
\subsection{Solution Existence}
We now turn to the fundamental task of determining whether an abduction problem has a~solution. To establish the complexity of deciding whether $\Sol(\mathbb{P})=\varnothing$, we provide reductions from classical abduction problems. We adapt the definition of \emph{classical abduction problems} from~\cite{EiterGottlob1995,CreignouZanuttini2006} to our notation.
\begin{definition}\label{def:CPLabductiveproblem}
A~\emph{classical abduction problem} is a~tuple $\mathbb{P}=\langle\Gamma,\chi,\Hmsf\rangle$ s.t.\ $\Gamma\cup\{\chi\}\subseteq\LBD$ and $\Hmsf$ is a~set of propositional literals. 
\begin{itemize}
\item A~\emph{solution} of $\mathbb{P}$ is a~conjunction $\tau$ of literals from $\Hmsf$ such that 
$\Gamma,\tau\models_\CPL\psi$ and $\Gamma,\tau\not\models_\CPL\bot$.

\item A~solution $\tau$ is \emph{proper} if $\tau\not\models_\CPL\psi$.

\item A~proper solution $\tau$ is \emph{theory-minimal} if there is no proper solution $\phi$ s.t.\ $\Gamma,\tau\models_\CPL\phi$ and $\Gamma,\phi\not\models_\CPL\tau$.

\item A~proper solution $\tau$ is \emph{$\models_\CPL$-minimal} if there is no proper solution $\phi$ s.t.\ $\tau\!\models_\CPL\!\phi$ and $\CPL\not\models\phi\leftrightarrow\tau$.
\end{itemize}
\end{definition}
\begin{theorem}\label{theorem:solutionexistence}
It is $\Sigma^\Pmsf_2$-complete to decide whether a 
$\BD$ abduction problem has a~(proper) $\Ltriangle$- or $\Lcirc$-solution.
\end{theorem}
\begin{proof}
Membership follows immediately from Theorems~\ref{theorem:anyproperBDtrianglesolution} and~\ref{theorem:anyproperBDcircsolution}. To show $\Sigma^\Pmsf_2$-hardness, we reduce solution existence for classical abduction problems $\mathbb{P}_\cl=\langle\Gamma_\cl,\chi_\cl,\Hmsf\rangle$ of the following form:\footnote{We write $\neg r\leftrightarrow r'$ as a~shorthand for $(r\wedge\neg r')\vee(\neg r\wedge r')$.}
\begin{align}
\Gamma_\cl&=\{\neg\phi\vee(p\wedge\tau),\neg p\vee\tau\}\cup\nonumber\\&\quad\hspace{.45em}\{\neg r\leftrightarrow r'\mid r\in\Prop(\phi)\setminus\Prop(p\wedge\tau)\}\nonumber\tag{$p\notin\Prop(\phi\wedge\tau)$, $\tau$ is a~term}\\
\chi_\cl&=p\wedge\tau\nonumber\\
\Hmsf
&=\{r\mid r\!\in\!\Prop(\phi)\!\setminus\!\Prop(p\!\wedge\!\tau)\}\!\cup\!\{r'\mid\neg r\!\leftrightarrow\!r'\!\in\!\Gamma_\cl\}
\label{equ:solutionexistenceCPL}
\end{align}
By~\cite[Theorem~4.2]{EiterGottlob1995}, determining the existence of classical solutions for these problems is $\Sigma^\Pmsf_2$-hard. We reduce $\mathbb{P}_\cl$ to $\mathbb{P}^\four=\langle\Gamma^\four,\chi^\four,\Hmsf\rangle$, where:
\begin{align}
\Gamma^\four&=\Gamma_\cl\cup\{q\vee\neg q\mid q\in\Prop[\Gamma_\cl]\}\nonumber\\
\chi^\four&=\chi_\cl\vee\bigvee\limits_{q\in\Prop[\Gamma_\cl]}(q\wedge\neg q)
\label{equ:solutionexistenceBD}
\end{align}

First let $\sigma$ be a~solution of $\mathbb{P}_\cl$. It is immediate from~\eqref{equ:solutionexistenceCPL} that $\sigma$ is a~\emph{proper} solution because we cannot use variables occurring in $\chi_\cl$. Moreover, by Proposition~\ref{prop:CPLtoBD}, we have that $\sigma$ is a~\emph{proper} solution of $\mathbb{P}^\four$. And since $\sigma\in\LBD$, it is both an $\Ltriangle$- and $\Lcirc$-proper solution.

For the converse, let $\sigma'$ be a~solution of $\mathbb{P}^\four$. As $\Hmsf$ contains only positive literals and no variables from $\Prop(p\wedge\tau)$, it follows that 
$\Gamma_\cl,\sigma' \not \models_\CPL\bot$ and 
$\sigma' \not \models_\CPL p\wedge\tau
$. Moreover, applying Proposition~\ref{prop:CPLtoBD}, we obtain that $\Gamma_\cl,\sigma'\models_\CPL\chi_\cl$. This shows that $\sigma'$ is a~ proper solution of $\mathbb{P}_\cl$. 
\end{proof}
\subsection{Relevance and Necessity of Hypotheses}
Two other natural reasoning tasks that arise in the context of abduction are the recognition of which hypotheses are \emph{relevant}, in the sense that they belong to at least one (minimal) solution, and which are \emph{necessary} (or \emph{indispensable}), as they occur in every (minimal) solution. Both of these decision problems have been investigated in the case of $\CPL$ abduction, see \cite{EiterGottlob1995}. 

The following theorem shows that the complexity of relevance and necessity w.r.t.\ (proper) solutions and $\models_\BD$-minimal solutions coincides with the complexity of the analogous problems for ($\subseteq$-minimal) solutions in $\CPL$.
\begin{theorem}\label{theorem:relevancecomplexity}
It is $\Sigma^\Pmsf_2$-complete (resp.\ $\Pi^\Pmsf_2$-complete) to decide, given a~$\BD$ abduction problem $\mathbb{P}=\langle\Gamma,\chi,\Hmsf\rangle$ and $h \in \Hmsf$, whether $h$ is relevant (resp.\ necessary) w.r.t.\ $\Sol(\mathbb{P})$. 
The same holds for relevance and necessity w.r.t.\ $\PropSol(\mathbb{P})$ and $\BDminSol(\mathbb{P})$. 
\end{theorem}
\begin{proof}
Membership is straightforward since solution recognition is $\conp$-complete for proper $\Ltriangle$-solutions and $\DP$-complete for proper $\Lcirc$-solutions and it suffices to guess a~proper solution containing (or omitting) $h$ and verify it. 

For the hardness results for $\models_\BD$-minimal solutions, we construct a~reduction from the class of $\BD$ abduction problems presented in~\eqref{equ:solutionexistenceBD} and adapt the approach from~\cite{EiterGottlob1995}. Namely, we let $\mathbb{P}^\four=\langle\Gamma^\four,\chi^\four,\Hmsf\rangle$ be as in~\eqref{equ:solutionexistenceBD} and pick fresh variables $r$, $r'$, and $r''$. Now set $\Xi=\Prop[\Gamma^\four]\cup\{r,r',r''\}$ and define $\mathbb{P}^\mathsf{rd}=\langle\Gamma^\mathsf{rd},\chi^\mathsf{rd},\Hmsf^\mathsf{rd}\rangle$ as follows.
\begin{align}
\Gamma^\mathsf{rd}&=\{\neg r\vee\psi\mid\psi\in\Gamma^\four\}\cup\{s\vee\neg s\mid s\in\Xi\}\cup\nonumber\\
&\quad\hspace{.45em}\{\neg r'\vee(p\wedge\tau),\neg r\vee\neg r',\neg(r\vee r')\vee r''\}\nonumber\\
\chi^\mathsf{rd}&=(p\wedge r''\wedge\tau)\vee\bigvee\limits_{s\in\Xi}(s\wedge\neg s)\nonumber\\
\Hmsf^\mathsf{rd}&=\Hmsf^\four\cup\{r,r'\}
\label{equ:relevanceBD}
\end{align}
Now let $\PropSol(\mathbb{P}^\four)$ be the set of all proper solutions of $\mathbb{P}^\four$ and recall that $\Sol(\mathbb{P}^\four)=\PropSol(\mathbb{P}^\four)$. It is clear that
\begin{align*}
\PropSol(\mathbb{P}^\mathsf{rd})&=\Sol(\mathbb{P}^\mathsf{rd})\\
\PropSol(\mathbb{P}^\mathsf{rd})&=\left\{\varrho\!\wedge\!r\!\mid\!\varrho\!\in\!\Sol(\mathbb{P}^\four)\right\}\!\cup\!\left\{\varrho'\!\wedge\!r'\!\mid\!\exists\Hmsf'\!\subseteq\!\Hmsf\!:\!\varrho'\!\!=\!\!\bigwedge\limits_{l\in\Hmsf}\!\!l\right\}
\end{align*}
and that $\mathbb{P}^\four$ has (proper) solutions iff $r$ is relevant and $r'$ is not necessary. 

To show hardness w.r.t.\ $\BDminSol(\mathbb{P})$, it suffices to observe that $r$ is relevant to $\mathbb{P}^\mathsf{rd}$ iff it is relevant w.r.t.\ $\models_\BD$-minimal solutions. 
Similarly, $r'$ is (not) necessary in $\mathbb{P}^\mathsf{rd}$ iff it is (not) necessary w.r.t.\ $\models_\BD$-minimal solutions.
\end{proof}

Finally, we present the following upper bound for the recognition of relevant and necessary hypotheses w.r.t.\ theory-minimal solutions. The proof is an easy consequence of Theorem~\ref{theorem:theoryminimalrecognition}. To the best of our knowledge, no analogous problem has been considered for $\CPL$ abduction problems. 
\begin{theorem}\label{theorem:theoryrelevancerecognition}
It is in $\Sigma^\Pmsf_3$ (resp.\ $\Pi^\Pmsf_3$)
to decide, given 
a~$\BD$ abduction problem $\mathbb{P}=\langle\Gamma,\chi,\Hmsf\rangle$ and $h \in \Hmsf$, whether $h$ is relevant (resp.\ necessary) w.r.t.\ $\ThminSol(\mathbb{P})$.

\end{theorem}
\section{Generating Solutions to $\BD$ Abduction Problems by Reduction to $\CPL$
\label{sec:problemembeddings}}
In this section, we show how to apply classical consequence-finding procedures to generate solutions for $\BD$ abduction problems, by reducing $\BD$ abduction to $\CPL$ abduction. 

We first observe that Lemma~\ref{lemma:BDtriangletoCPL} allows us to faithfully translate abduction problems with $\Ltriangle$-solutions into~$\CPL$.
\begin{theorem}\label{theorem:BDtriangletoCPLabduction}
Let $\mathbb{P}=\langle\Gamma,\chi,\Hmsf\rangle$ be a~$\BD$ abduction problem. Then $\phi$ is a~($\models_\BD$-, theory-minimal, proper) $\Ltriangle$-solution of $\mathbb{P}$ iff $\phi^\cl$ is a~($\models_\CPL$-, theory-minimal, proper) solution of $\mathbb{P}^\cl_\triangle=\langle\Gamma^\cl,\chi^\cl,\Hmsf^\cl_\triangle\rangle$ with $\Hmsf^\cl_\triangle=\{p^+\mid p\in\Hmsf\}\cup\{p^-\mid \neg p\in\Hmsf\}\cup\{{\sim}l\mid\neg\triangle l\in\Hmsf\}$.
\end{theorem}

The translation of $\BDcirc$ abduction into $\CPL$ abduction is, however, more complicated.
\begin{definition}\label{def:atomiccirctermtoCPL}
Let $\phi\!=\!\bigwedge\limits^{m}_{i=1}p_i\!\wedge\!\bigwedge\limits^{m'}_{i'=1}\!\neg p_{i'}\!\wedge\!\bigwedge\limits^{n}_{j=1}\!\circ p_j\!\wedge\!\bigwedge\limits^{n'}_{j'=1}\!\bullet p_{j'}$ be an atomic $\Lcirc$-term, and $X\supseteq\Prop(\phi)$ be a finite set of propositional variables. The \emph{classical counterpart} of $\phi$ relative to $X$, denoted $\phi^\circ_X$, is defined as follows:
\begin{align}
\phi^{\sim}&=\bigwedge\limits^{m}_{i=1}p^+_i\!\wedge\!\bigwedge\limits^{m'}_{i'=1}p^-_{i'}\!\wedge\!\bigwedge\limits^{n}_{j=1}p^\circ_j\!\wedge\!\bigwedge\limits^{n'}_{j'=1}\!\!{\sim}p^\circ_{j'}\nonumber\\
\phi^{\leftrightarrow}_X&=\bigwedge\limits_{q\in X}\!\!({\sim}q^\circ\!\leftrightarrow\!(q^+\!\leftrightarrow\!q^-
))\nonumber\\
\phi^\circ_X&=\phi^{\sim}\wedge\phi^\leftrightarrow_X
\label{equ:atomictermcl}
\end{align}
where the $p_i^+$, $p_i^-$, and $q^\circ$s are fresh variables (not in $X$). 
\end{definition}
\begin{lemma}\label{lemma:BDcirctoCPLatomicterms}
Let $\chi,\psi\in\LBD$, $\Xi=\Prop(\chi)\cup\Prop(\psi)$, and $\phi$ be an \emph{atomic} $\Lcirc$-term s.t.\ $\Prop(\phi)\subseteq\Xi$. Then 
\begin{align*}
\phi,\chi\models_\BD\psi&\text{ iff }\phi^\circ_\Xi,\chi^\cl\models_\CPL\psi^\cl
\end{align*}
\end{lemma}

We can now use Lemma~\ref{lemma:BDcirctoCPLatomicterms} to construct a~faithful embedding of $\BD$ abduction problems with $\Lcirc$-solutions into classical abduction problems. Note that the size of $\phi^\circ_X$ is linear in the cardinality of $X$. This, however, is only possible because $\phi$ is an \emph{atomic term}. Furthermore, since atomic $\Lcirc$-terms are not always translated into conjunctions of literals, we need to modify the statement of Theorem~\ref{theorem:BDtriangletoCPLabduction}. Moreover, the resulting embedding preserves only theory-minimality since $\phi^{\sim}$ does not govern the interaction between $p^\circ$s, $p^+$s, and $p^-$s.
\begin{theorem}\label{theorem:BDcircCPLabduction}
Let $\mathbb{P}=\langle\Gamma,\chi,\Hmsf\rangle$ be a~$\BD$ abduction problem and $\Xi=\Prop[\Gamma\cup\{\chi\}]$. Then $\phi$~is a (theory-minimal, proper) $\Lcirc$-solution of $\mathbb{P}$ iff $\phi^{\sim}$ is a~(theory-minimal, proper) solution of $\mathbb{P}^\cl_\circ\!=\!\left\langle\Gamma^\cl\!\cup\!\left\{\phi^{\leftrightarrow}_{\Xi}\right\},\phi^{\leftrightarrow}_{\Xi}\rightarrow\chi^\cl,\Hmsf^\cl_\circ\right\rangle$ with $\Hmsf^\cl_\circ=\{p^+\mid p\in\Hmsf\}\cup\{p^-\mid\neg p\in\Hmsf\}\cup\{p^\circ\mid\circ p\in\Hmsf\}\cup\{{\sim}p^\circ\mid\bullet p\in\Hmsf\}$.
\end{theorem}

Theorems~\ref{theorem:BDtriangletoCPLabduction} and~\ref{theorem:BDcircCPLabduction} show that we can use classical techniques of abductive reasoning (such as the ones based upon consequence finding, cf.\ \cite{Inoue1992,delVal2000,Marquis2000HDRUMS,Inoue2002}) to solve $\BD$ abductive problems. Let us now illustrate how our translations work.
\begin{example}\label{example:translation}
Recall Example~\ref{example:impossible1}. We consider two $\BD$ problems: $\mathbb{P}_\triangle=\langle\Gamma,\chi,\Hmsf_\triangle\rangle$ and $\mathbb{P}_\circ=\langle\Gamma,\chi,\Hmsf_\circ\rangle$, where:
\begin{align*}
\Gamma&=\{p\!\vee\!q,\neg p,\neg q\}&\chi&=q&\Hmsf_\triangle&=\{p,\!\neg p,\!\neg\triangle p,\!\neg\triangle\neg p\}\\
&&&&\Hmsf_\circ&=\{p,\neg p,\circ p,\bullet p\}
\end{align*}
Applying Theorems~\ref{theorem:BDtriangletoCPLabduction} and~\ref{theorem:BDcircCPLabduction}, we obtain the following classical problems $\mathbb{P}_\triangle^\cl=\langle\Gamma^\cl,\chi^\cl,\Hmsf_\triangle^\cl\rangle$ and $\mathbb{P}_\circ^\cl=\langle\Gamma^\circ,\chi^\cl,\Hmsf^\cl_\circ\rangle$, where:
\begin{align*}
\Gamma^\cl&=\{p^+\vee q^+,p^-,q^-\}&
\chi^\cl&=q^+\\
\Hmsf_\triangle^\cl&=\{p^+,p^-,{\sim}p^+,{\sim}p^-\}\\
\Gamma^\circ&=\Gamma^\cl\cup\left\{({\sim}r^\circ\!\leftrightarrow\!(r^+\!\leftrightarrow\!r^-))\mid r\in\{p,q\}\right\}\\
\Hmsf_\circ^\cl&=\{p^+,p^-,p^\circ,{\sim}p^\circ\}
\end{align*}
Now if we apply classical consequence-finding procedures, we need to look for clauses entailed by $\Gamma_\triangle^\cl\cup\{{\sim}\chi^\cl\}$ and $\Gamma_\circ^\cl\cup\{{\sim}\chi^\cl\}$. For $\Ltriangle$-solutions, we can take (the negation of) \emph{any clause}. For $\Lcirc$-solutions, the clauses must contain only \emph{negative} occurrences of $p^+$ and~$p^-$.

It is easy to check that ${\sim}p^+$ and $p^\circ$ are theory-minimal \emph{classical} solutions of $\mathbb{P}_\triangle^\cl$ and $\mathbb{P}_\circ^\cl$. They correspond to $\neg\triangle p$ and $\circ p$, respectively, which are (as expected) theory-minimal $\Ltriangle$- and $\Lcirc$- solutions of $\mathbb{P}^\triangle$ and $\mathbb{P}^\circ$. 
\end{example}

We finish the section by noting that in general, $\Ltriangle$-solutions are not uniquely generated from classical clauses (cf.~$p^\cl=(\triangle p)^\cl=p^+$) and that by Proposition~\ref{prop:triangleequivalence}, $\phi\simeq\phi^\flat$. Thus, for practical purposes, it makes sense to convert classical clauses to $\Ltriangle$ solutions $\tau$ s.t.\ $\tau=\tau^\flat$.
\section{Discussion and Future Work\label{sec:conclusion}}
We have studied abductive reasoning in the four-valued paraconsistent logic $\BD$, motivating and comparing $\Ltriangle$- and $\Lcirc$-solutions. Our complexity analysis (Table~\ref{tab:complexitysolutions}) provides an almost complete picture of the complexity of the main decision problems related to abduction.
%
In particular, we established that the complexity of solution existence in $\BDtriangle$ and $\BDcirc$ is not higher than in the classical case~\cite{EiterGottlob1995,CreignouZanuttini2006,PichlerWoltran2010,PfandlerPichlerWoltran2015}. 
Moreover, by exhibiting reductions of abduction in $\BDtriangle$ and $\BDcirc$ to abduction in $\CPL$, we have shown that existing procedures for generating abductive solutions in classical logic can be employed for paraconsistent abduction. 

A few questions remain open. First, we do not know the exact complexity of theory-minimal solution recognition and relevance. One way to approach this would be to establish the complexity of closely related notion of theory prime implicants in $\CPL$~\cite{Marquis1995}. 
There is also the question of how to embed $\BD$ abduction problems with $\Lcirc$-solutions into classical problems while preserving $\models_\BD$-minimal solutions (recall that Theorem~\ref{theorem:BDcircCPLabduction} only preserves theory-minimal solutions). 
Also, since some $\BD$ abduction problems can be solved by arbitrary $\Lcirc$-terms but not atomic ones, it would be interesting to explore the computational properties of $\Lcirc$-solutions based upon non-atomic $\Lcirc$-terms. 

A more general direction for future work is to consider abduction in expansions of $\BD$. Of particular interest are \emph{functionally complete} expansions of $\BD$, e.g., the bi-lattice language expansion or an expansion with a~‘quarter-turn’ connective from~\cite{Ruet1996} (cf.~\cite{OmoriSano2015} for more details). An important technical question would be whether all $\Ltriangle$- and $\Lcirc$-solutions can be represented as terms in functionally complete languages. A~further challenge is to come up with an intuitive natural-language interpretation of literals and terms in such languages.
Another option is to consider theories containing implicative formulas (cf.~\cite{OmoriWansing2017} for details). This would allow us to define Horn-like fragments of the languages in question. As solution existence in \emph{classical} Horn abduction problems is $\np$-complete~\cite{CreignouZanuttini2006}, it makes sense to check whether abduction in Horn $\BD$ is also simpler than in the general case.

Additionally, we plan to consider \emph{modal} expansions of $\BD$. Abduction in classical modal logic is well-researched. In particular, \cite{Levesque1989} and~\cite{SakamaInoue2016} study abduction in classical epistemic and doxastic logics; \cite{MayerPirri1995} apply tableaux procedures to solution generation in $\mathbf{K}$, $\mathbf{D}$, $\true$, and $\mathbf{S4}$. \cite{Bienvenu2009} compares different definitions of prime implicates closely related to abductive solutions and provides complexity results and algorithms for prime implicate recognition and generation in multimodal $\mathbf{K}_n$; \cite{Nepomuceno-FernandezSoler-ToscanoVelazquez-Quesada2017} consider abduction in dynamic epistemic logic. Modal expansions of $\BD$ are also well known (cf.~\cite{Priest2008RSL} and~\cite{Drobyshevich2020}). There are also public announcement~\cite{Rivieccio2014} and dynamic~\cite{Sedlar2016} $\BD$ logics. Thus, it is natural to consider abductive reasoning in \emph{modal paraconsistent} framework and see whether classical decision procedures and complexity results can be transferred there.
\section*{Acknowledgements}
The authors were supported by the ANR AI Chair INTENDED (ANR-19-CHIA-0014) ($1^\textnormal{st}$ and $3^\textnormal{rd}$), JSPS KAKENHI (JP21H04905) and the JST CREST (JPMJCR22D3)~($2^\textnormal{nd}$), and an NII MOU grant  ($3^\textnormal{rd}$).


\bibliographystyle{kr}
\bibliography{kr-sample}
\newpage
\appendix
\section{Proofs of Section~\ref{sec:BD}}
\subsection{Proof of Proposition~\ref{prop:contraposition}\label{ssec:contrapositionproof}}
Let $\phi,\chi\in\Lcirc$ and $\varrho,\sigma,\tau\in\Ltriangle$. Then the following statements hold.
\begin{enumerate}
\item $\phi\models_\BD\chi$ iff $\neg\chi\models_\BD\neg\phi$.
\item $\varrho,\sigma\models_\BD\tau$ iff $\varrho,\neg\triangle\tau\models_\BD\neg\triangle\sigma$.
\item $\varrho,\sigma\models_\BD\tau$ iff $\varrho\models_\BD\neg\triangle\sigma\vee\tau$.
\end{enumerate}
\begin{proof}
We begin with Statement 1. Let $v$ be a~$\BD$ valuation. Define $v^\partial$ as follows: $v^\partial(p)=\neither$ if $v(p)=\both$; $v^\partial(p)=\both$ if $v(p)=\neither$; $v^\partial(p)=v(p)$ otherwise. It is easy to check by induction on $\psi\in\Lcirc$ that $v^\partial(\psi)=\neither$ if $v(\psi)=\both$; $v^\partial(\psi)=\both$ if $v(\psi)=\neither$; $v^\partial(\psi)=v(\psi)$ otherwise.

From here, it follows immediately that the contraposition holds. Indeed, assume that $v(\neg\chi)\in\{\true,\both\}$ and $v(\neg\phi)\notin\{\true,\both\}$. Hence, $v(\chi)\in\{\false,\both\}$ and $v(\phi)\notin\{\false,\both\}$. Consider the case when $v(\phi)=\true$ and $v(\chi)=\both$. Then, $v^\partial(\phi)=\true$ but $v^\partial(\chi)=\neither$, whence $\phi\not\models_\BD\chi$, as required. Other cases can be tackled similarly.

For Statement 2, assume that $v(\varrho)\in\{\true,\both\}$, $v(\neg\triangle\tau)=\true$ and $v(\neg\triangle\sigma)=\false$. It is clear that $v(\sigma)\in\{\true,\both\}$ but $v(\tau)\in\{\neither,\false\}$. Hence, $\varrho,\sigma\not\models_\BD\tau$, as required. The converse direction can be shown in the same manner.

Finally, Statement 3 can be verified by a~routine check.
\end{proof}
\subsection{Proof of Proposition~\ref{prop:triangleequivalence}\label{ssec:triangleequivalenceproof}}
We show that $\phi\simeq\phi^\flat$ for every $\phi\in\Lcirctriangle$.
\begin{proof}
Since $\circ$ is definable via $\triangle$, we can assume that $\phi$ is in $\DNF$. Now, let $l$, $l'$, and $l''$ be literals s.t.\ $\tau=\bigwedge l\wedge\bigwedge\triangle l\wedge\bigwedge\neg\triangle l''$. Using~\eqref{equ:triangleequivalence}, we obtain that $\tau\equiv\bigwedge l\wedge\triangle\left(\bigwedge l\right)\wedge\bigwedge\neg\triangle l''$. It is now easy to check that for every $\phi,\chi\in\Ltriangle$ and every $\circledast\in\{\wedge,\vee\}$, it holds that $\phi\!\circledast\!\triangle\chi\!\simeq\!\phi\!\circledast\!\chi$. Thus, $\tau\!\simeq\!\tau^\flat$ for every disjunct of~$\phi$.

It remains to show that if $\tau_1\simeq\tau'_1$ and $\tau_2\simeq\tau'_2$, then $\tau_1\circledast\tau_2\simeq\tau'_1\circledast\tau'_2$. For $\circledast=\wedge$, we have the following.
\begin{align*}
v(\tau_1\wedge\tau_2)\!\in\!\{\true,\both\}&\text{ iff }v(\tau_1)\!\in\!\{\true,\both\}\text{ and }v(\tau_2)\!\in\!\{\true,\both\}\\
&\text{ iff }v(\tau'_1)\!\in\!\{\true,\both\}\text{ and }v(\tau'_2)\!\in\!\{\true,\both\}\tag{by definition of $\simeq$}\\
&\text{ iff }v(\tau'_1\wedge\tau'_2)\!\in\!\{\true,\both\}
\end{align*}
And similarly for $\circledast=\vee$:
\begin{align*}
v(\tau_1\vee\tau_2)\!\in\!\{\true,\both\}&\text{ iff }v(\tau_1)\!\in\!\{\true,\both\}\text{ or }v(\tau_2)\!\in\!\{\true,\both\}\\
&\text{ iff }v(\tau'_1)\!\in\!\{\true,\both\}\text{ or }v(\tau'_2)\!\in\!\{\true,\both\}\tag{by definition of $\simeq$}\\
&\text{ iff }v(\tau'_1\vee\tau'_2)\!\in\!\{\true,\both\} \qedhere
\end{align*}
\end{proof}
\subsection{Proof of Proposition~\ref{prop:CPLtoBDtrianglecirc}\label{ssec:CPLtoBDtrianglecircproof}}
We show that $\CPL\models\phi$ iff $\BD\models\phi^\triangle$ iff $\BD\models\phi^\circ$.
\begin{proof}
We start with a~proof that $\phi$ is $\CPL$-valid iff $\phi^\circ$ is $\BD$-valid. Observe that $v(\phi^\circ)\in\{\true,\false\}$. Now let $\mathbf{v}$ be a~\emph{classical} valuation s.t.\ $\mathbf{v}(\phi)=\false$. We construct $v_\mathbf{v}$ as follows: $v_\mathbf{v}(p)=\true$ iff $\mathbf{v}(p)=\true$ and $v_\mathbf{v}(p)=\neither$ otherwise. It can be established by a~straightforward induction that $v_\mathbf{v}(\phi^\circ)=\mathbf{v}(\phi)$. For the converse direction, we let $v$ be a~$\BD$ valuation s.t.\ $v(\phi^\circ)=\false$ and define $\mathbf{v}_v(p)=v(\circ p)$. Again, it is easy to check that $v(\phi^\circ)=\mathbf{v}_v(\phi)$.

For $\phi$ and $\phi^\triangle$, we can proceed similarly. Let $\mathbf{v}(\phi)=\false$ and define $v_\mathbf{v}(p)=\mathbf{v}(p)$. It is clear that $\mathbf{v}(\phi)=v_\mathbf{v}(\phi^\triangle)$. For the converse direction, let $v(\phi^\triangle)=\false$ and define $\mathbf{v}_v(p)=v(\triangle p)$. Again, it is easy to verify that $v(\phi^\triangle)=\mathbf{v}_v(\phi)$.
\end{proof}
\subsection{Proof of Proposition~\ref{prop:CPLtoBD}\label{ssec:CPLtoBDproof}}
We show that $\phi\models_\CPL\chi$ iff $\phi\wedge\bigwedge\limits^{n}_{i=1}(p_i\vee\neg p_i)\models_\BD\chi\vee\bigvee\limits^{n}_{i=1}(p_i\wedge\neg p_i)$.
\begin{proof}
The ‘only if’ direction is straightforward as $p\wedge\neg p$ is classically unsatisfiable and $p\vee\neg p$ is classically valid. For the ‘if’ direction, let $v\left(\phi\wedge\bigwedge\limits^{n}_{i=1}(p_i\vee\neg p_i)\right)\in\{\true,\both\}$ and $v\left(\chi\vee\bigvee\limits^{n}_{i=1}(p_i\wedge\neg p_i)\right)\in\{\neither,\false\}$. It suffices to show that $v(p_i)\in\{\true,\false\}$ for all $p_i$'s. Indeed, otherwise, if $v(p_i)=\both$ for some $p_i$, then $v\left(\chi\vee\bigvee\limits^{n}_{i=1}(p_i\wedge\neg p_i)\right)\in\{\true,\both\}$, contrary to the assumption. And if $v(p_i)\!=\!\neither$ for some $p_i$, then $v\left(\phi\!\wedge\!\bigwedge\limits^{n}_{i=1}(p_i\!\vee\!\neg p_i)\right)\!\in\!\{\neither,\!\false\}$, contrary to the assumption.
\end{proof}
\section{Proofs of Section~\ref{sec:complexity}}
\subsection{Proof of Theorem~\ref{theorem:atomiccirctermPtime}\label{ssec:atomiccirctermPtimeproof}}
We prove the following lemma.
\begin{lemma}\label{lemma:atomiccirctermentailmentcriterion}
For any \emph{$\BD$-satisfiable} atomic $\Lcirc$-terms $\sigma$ and $\sigma'$, it holds that $\sigma\models_\BD\sigma'$ iff
\begin{enumerate}
\item[(a)]every literal $\circ p$ that occurs in $\sigma'$ also occurs in $\sigma$;
\item[(b)]for every literal $\bullet p$ occurring in $\sigma'$, $\bullet p$ occurs in $\sigma$ or $p$ and $\neg p$ occur in $\sigma$;
\item[(c)] for every propositional literal $l$ occurring in~$\sigma'$, $l$ occurs in $\sigma$ or $\overline{l}$ and $\bullet l$ occur in $\sigma$.
\end{enumerate}
\end{lemma}
\begin{proof}
The ‘only if’ direction is evident since $p\wedge\neg p\models_\BD\bullet p$ and $l\wedge\bullet l\models_\BD\overline{l}$. For the ‘if’ direction, we reason by cases. Let $v$ be any satisfying valuation for $\sigma$. We are going to modify $v$ so that $v(\sigma)\in\{\true,\both\}$ but $v(\sigma')\notin\{\true,\both\}$. We consider three cases. (1) If there is some $\circ p$ in $\sigma'$ but not in $\sigma$, we set $v(p)=\both$. This gives $v(\sigma')=\false$. (2) If $\bullet p$ occurs in $\sigma'$ but neither $\bullet p$ nor $p$ together with $\neg p$ do, we set $v(p)=\true$ if $p$~occurs in $\sigma$ and $v(p)=\false$, otherwise. Again, we have $v(\sigma')=\false$. (3) Finally, if there is some $l$ in $\sigma'$ s.t.\ neither $l$ nor $\overline{l}$ together with $\bullet l$ occur in $\sigma$, we set $v(p)=\false$ if $l=p$ and $v(p)=\true$ if $l=\neg p$.
\end{proof}

Since it takes polynomial time to check whether (a)--(c) hold, Theorem~\ref{theorem:atomiccirctermPtime} follows.
\subsection{Proof of Lemma~\ref{lemma:BDtriangletoCPL}\label{ssec:BDtriangletoCPLproof}}
\begin{proof}
Let $\phi\not\models_\BD\chi$ and $v$ be a~$\BD$ valuation s.t.\ $v(\phi)\in\{\true,\both\}$ and $v(\chi)\notin\{\true,\both\}$. We set
\begin{align}\label{equ:vcl}
v^\cl(p^+)\!=\!\true&\text{ iff }v(p)\!\in\!\{\true,\both\}\nonumber\\v^\cl(p^-)\!=\!\true&\text{ iff }v(p)\!\in\!\{\false,\both\}
\end{align}
One can now easily check by induction that for every $\psi\in\Ltriangle$, $v(\psi)\in\{\true,\both\}$ iff $v^\cl(\psi^\cl)=\true$.

For the converse direction, we let $\mathbf{v}$ be a~\emph{classical} valuation s.t.\ $\mathbf{v}(\phi)=\true$ and $\mathbf{v}(\chi)=\false$. Now set
\begin{align}\label{equ:vfour}
\mathbf{v}^\four&=
\begin{cases}
\true&\text{ iff }\mathbf{v}(p^+)=\true\text{ and }\mathbf{v}(p^-)=\false\\
\both&\text{ iff }\mathbf{v}(p^+)=\true\text{ and }\mathbf{v}(p^-)=\true\\
\neither&\text{ iff }\mathbf{v}(p^+)=\false\text{ and }\mathbf{v}(p^-)=\false\\
\false&\text{ iff }\mathbf{v}(p^+)=\false\text{ and }\mathbf{v}(p^-)=\true
\end{cases}
\end{align}
Again, it is easy to check that $\mathbf{v}^\four(\psi^\cl)\in\{\true,\both\}$ iff $\mathbf{v}(\psi)=\true$. Note that since $\triangle p\simeq p$, the lack of injectivity of $^\cl$ is not problematic.
\end{proof}
\subsection{Proof of Theorem~\ref{theorem:triangletermsformulasPtime}}
We show that $\phi\models_\BD\chi$ can be decided in polynomial time when $\phi$ is a~\emph{satisfiable $\Ltriangle$-term} and $\chi\in\LBD$ is in $\NNF$.
\begin{proof}
By Lemma~\ref{lemma:BDtriangletoCPL}, we know that this is equivalent to $\phi^\cl\models_\BD\chi^\cl$. Now, we substitute every positive literal in $\phi^\cl$ with $\top$ and every negative with $\bot$ and apply reductions from~\eqref{equ:01reduction} to $\chi^\cl$. Note that the reductions can be conducted in polynomial time w.r.t.\ the length of $\chi$. Furthermore, it is clear that if $(\chi^\cl)^\sharp=\top$, then $\phi^\cl\models_\CPL\chi^\cl$. On the other hand, if $(\chi^\cl)^\sharp\neq\top$, then either $(\chi^\cl)^\sharp=\bot$ (in which case, the entailment evidently fails) or $(\chi^\cl)^\sharp=\psi$ for some ${\sim}$-free $\psi\in\LBD$. We construct a~falsifying valuation as follows: for every $p\in\Prop(\psi)$, we set $v(p)=\false$ and for every $q\in\Prop(\phi)$, we set $v(q)=\true$ iff ${\sim}q$ \emph{does not occur} in $\phi$. As reductions in~\eqref{equ:01reduction} preserve $\CPL$-equivalence, we have that $v(\phi^\cl)=\true$ and $v(\chi^\cl)=\false$. The result follows.
\end{proof}
\section{Proofs of Section~\ref{sec:solvingBDproblems}}
\subsection{Proof of Theorem~\ref{theorem:minimalrecognition}\label{ssec:minimalrecognitionproof}}
First, we need a~technical lemma.
\begin{lemma}\label{lemma:circprefixes}
Define
\begin{align*}
\sigma^\Vmc&=\{\true\both p\mid\text{only $p$ is in }\sigma\}\!\cup\!\{\false\both p\mid\text{only $\neg p$ is in }\sigma\}\cup\\
&\quad\hspace{.45em}\{\true p\mid p,\circ p\in\circLit(\sigma)\}\!\cup\!\{\false p\mid\neg p,\circ p\in\circLit(\sigma)\}\!\cup\\
&\quad\hspace{.45em}\{\both\neither p\mid\text{only $\bullet p$ is in }\sigma\}\!\cup\!\{\true\false p\mid\text{only $\circ p$ is in }\sigma\}\!\cup\\
&\quad\hspace{.45em}\{\both p\mid\text{any two out of }\{p,\neg p,\bullet p\}\text{ are in }\sigma\}
\end{align*}

Then if we treat prefixes $\overline{\Xmbf}$ and $\overline{\mathbf{Y}}$ as sets, it holds that
\begin{align*}
\sigma\models_\BD\sigma'&\text{ iff }\forall p\in\Prop:
\left(\begin{matrix}
\text{if }\overline{\mathbf{Y}}p\in\sigma'^\Vmc,\text{ then}\\
\exists\overline{\Xmbf}(\overline{\Xmbf}p\in\sigma^\Vmc~\&~\overline{\Xmbf}\subseteq\overline{\mathbf{Y}})
\end{matrix}\right)
\end{align*}
\end{lemma}
\begin{proof}
Immediately from Lemma~\ref{lemma:atomiccirctermentailmentcriterion} (cf.~Section~\ref{ssec:atomiccirctermPtimeproof}).
\end{proof}
\begin{proof}[Proof of Theorem~\ref{theorem:minimalrecognition}]
We start by showing membership in $\DP$ for the $\Ltriangle$ case.  
Let $\sigma$ be an $\Ltriangle$-term. We can w.l.o.g.\ assume (recall Proposition~\ref{prop:triangleequivalence}) that all $\triangle$'s in $\sigma$ occur under $\neg$, i.e., $\sigma=\sigma^\flat$. It is clear from Lemma~\ref{lemma:BDtriangletoCPL} that given two $\Ltriangle$-terms $\sigma^\flat$ and $\sigma'^\flat$, we have $\sigma^\flat\models_\BD\sigma'^\flat$ iff $\triangleLit(\sigma^\flat)\supseteq\triangleLit(\sigma'^\flat)$. We can now proceed as follows. Given $\sigma$, we make a call to an $\np$ oracle that guesses (linearly many) valuations $v_\mathsf{sat}$, $v_\mathsf{proper}$, and $v_l$ (for each literal $l$ of $\tau$) and verifies that these valuations witness that (i) $\Gamma,\sigma\not\models_\BD\bot$, (ii) $\sigma\not\models_\BD\chi$, and (iii) $\sigma\not\models_\BD\sigma^{-l}$ for every literal $l$ of $\sigma$, where $\sigma^{-l}$ being the result of deleting one $\Ltriangle$-literal from $\sigma$. At the same time (as we do not need the results of other guesses), we make a $\conp$ check that $\Gamma,\sigma\models_\BD\chi$. It follows from the definition of $\models_\BD$-minimal solutions that $\sigma$ is a $\models_\BD$-minimal solution iff the $\np$ and $\conp$ checks succeed.

Now let $\sigma$ be an atomic $\Lcirc$-term. We proceed in the same manner as with $\Ltriangle$-terms but now we cannot just remove literals because of situations when one term entails another even though they do not have common atomic literals as in $p\wedge\neg p\models_\BD\bullet p$. To circumvent this problem, we transform $\sigma$ into $\sigma^\Vmc$ (defined in Lemma \ref{lemma:circprefixes}) and note that $\sigma\equiv\sigma'$ iff $\sigma^\Vmc=\sigma'^\Vmc$. Now, given $\sigma=\bigwedge\limits^{m}_{i=1}l_i$, we only need to check terms that can be obtained by removing a~literal (there are $m$ of those) and terms corresponding to $\sigma'^\Vmc$ obtained from $\sigma^\Vmc$ by replacing one $\overline{\Xmbf}p$ with some $\overline{\Ymbf}p$ s.t.\ $\overline{\Xmbf}\subseteq\overline{\Ymbf}$ (there are at most $3m$ of those).

For $\DP$-hardness, we reduce the re\-cog\-ni\-tion of classical prime implicants which is $\DP$-com\-p\-lete due to~\cite[Proposition~115]{Marquis2000HDRUMS} to recognition of $\models_\BD$-minimal $\LBD$-solutions (thus, the bound will work for both $\Ltriangle$ and $\Lcirc$). Namely, let $\chi\in\LBD$ and $\tau$ be a~term. We show that $\tau$ is a~prime implicant of $\chi$ iff $\tau$ is a~$\models_\BD$-minimal solution of
\begin{align*}
\mathbb{P}&=\left\langle\left\{\bigwedge\limits_{p\in\Prop(\chi\wedge q)}\!\!\!\!\!\!\!\!(p\vee\neg p),q\right\},(\chi\wedge q)\vee\!\!\!\!\bigvee\limits_{p\in\Prop(\chi)}\!\!\!\!\!\!\!\!(p\wedge\neg p),\Hmsf\right\rangle
\end{align*}
with $\Hmsf=\{r\mid r\in\Prop(\chi)\}\cup\{\neg r\mid r\in\Prop(\chi)\}$ and $q\notin\Hmsf$. Note that all improper solutions of $\mathbb{P}$ are classically unsatisfiable.

Now let $\tau$ be a~classical prime implicant of $\chi$. Then we have that $\tau\consvDashCPL\chi$ and there is no $\tau'$ s.t.\ $\tau\models_\CPL\tau'$ and $\tau'\consvDashCPL\chi$. By Proposition~\ref{prop:CPLtoBD}, we have that $\tau$ is indeed a~solution of $\mathbb{P}$ (observe from Definition~\ref{def:BDsemantics} that $\tau$ is consistent with any $\LBD$-theory because it only contains $\neg$ and $\wedge$). It is a~proper solution because it is classically satisfiable and does not contain $q$. It remains to show that $\tau$ is a~$\models_\BD$-minimal solution. For this, assume for contradiction that there is some weaker proper solution $\tau'$. But $\tau'$ must be classically satisfiable, whence, $\tau\models_\CPL\tau'$ (because term entailment in $\BD$ and $\CPL$ coincide\footnote{Observe that $\tau\models_\BD\tau'$ iff $\Lit(\tau)\supseteq\Lit(\tau')$ iff $\tau\models_\CPL\tau'$.} for the case of \emph{classically satisfiable} $\LBD$-terms). Moreover, by Proposition~\ref{prop:CPLtoBD}, we would have that $\tau'\models_\CPL\chi$ which would contradict the assumption of $\tau$ being a~\emph{prime} implicant.

The converse direction can be shown similarly. Let $\tau$ be a~$\models_\BD$-minimal solution of $\mathbb{P}$. Again, by Proposition~\ref{prop:CPLtoBD}, we have that $\tau\models_\CPL\chi$ and that there is no other $\tau'$ s.t.\ $\tau'\models_\CPL\chi$ and $\tau\models_\BD\tau'$ for else we would have $\tau\models_\BD\tau'$ and $\tau'$ be a~proper solution of $\mathbb{P}$ (again, recall that all proper solutions of $\mathbb{P}$ are classically satisfiable).
\end{proof}
\subsection{Proof of Theorem~\ref{theorem:theoryminimalrecognition}\label{ssec:theoryminimalrecognitionproof}}

Let $\mathbb{P}=\langle\Gamma,\chi,\Hmsf\rangle$ be a~$\BD$ abduction problem, and $\sigma$ is the atomic $\Lcirc$- or $\Ltriangle$-term we wish to test. 
We outline a $\Sigma^\Pmsf_2$ procedure for 
the complementary problem of testing whether $\sigma$ is \emph{not} a theory-minimal $\Ltriangle$-solution. 
We first guess either ‘not a proper solution’ or another term $\sigma'$ in the considered language. 
In the former case, due to the $\conp$- / $\DP$-completeness of proper solution recognition (Theorems \ref{theorem:anyproperBDtrianglesolution} and \ref{theorem:anyproperBDcircsolution}), 
we can make (one or two) calls to an $\np$ oracle to verify that
$\sigma$ is indeed not a proper solution (in which case we return yes). 
In the latter case, we can make a few calls to an $\np$ oracle to verify that 
$\sigma'$ is a proper solution such that $\Gamma, \sigma \models_\BD \sigma'$ and $\Gamma, \sigma' \not \models_\BD \sigma$ 
(returning yes if the calls show this to be the case). 
It is easy to see that the procedure we have just described will return yes 
iff the input $\sigma$ is not a theory-minimal solution, which yields the desired
membership in $\Pi^\Pmsf_2$ for the original problems. 

\section{Proofs of Section~\ref{sec:problemembeddings}}
\subsection{Proof of Lemma~\ref{lemma:BDcirctoCPLatomicterms}\label{ssec:BDcirctoCPLatomictermsproof}}
Let $\chi,\psi\in\LBD$ and $\phi$ be an \emph{atomic} $\Lcirc$-term s.t.\ $\Prop(\phi)\subseteq\Prop(\chi)\cup\Prop(\psi)$ and $\Xi=\Prop(\chi)\cup\Prop(\psi)$. Then
\begin{align*}
\phi,\chi\models_\BD\psi&\text{ iff }\phi^\circ_\Xi,\chi^\cl\models_\CPL\psi^\cl
\end{align*}
\begin{proof}
For the ‘if’ direction, assume that $\phi,\chi\not\models_\BD\psi$ and let $v$ be a~$\BD$-valuation that falsifies this entailment. We show that $v^\cl$ (recall~\eqref{equ:vcl}) falsifies the classical entailment. It is clear that $v^\cl(\chi^\cl)=\true$ and $v^\cl(\psi^\cl)=\false$, so, it remains to show that $v^\cl(\phi^\cl)=\true$. First, we set
\begin{align}\label{equ:vclpcirc}
\forall q\in\Xi:v^\cl(q^\circ)=\true&\Leftrightarrow v(q)\in\{\true,\false\}
\end{align}
Thus, $v({\sim}q^\circ\leftrightarrow(q^+\leftrightarrow q^-))=\true$. Now observe that since $v(\phi)\in\{\true,\both\}$, we have that $v(p_i)\in\{\true,\both\}$ and $v(p_{i'})\in\{\false,\both\}$ for $i\leq m$ and $i'\leq m'$. Thus, $v^\cl(p^+_i)=\true$ and $v^\cl(p^-_{i'})=\false$. In addition, we have that $v(p_j)\in\{\true,\false\}$ and $v(p_{j'})\in\{\both,\neither\}$ for $j\leq n$ and $j'\leq n'$. Thus, using~\eqref{equ:vclpcirc}, we obtain that $v^\cl(p^\circ_j)=\true$ and $v^\cl(p^\circ_{j'})=\false$ which gives us $v^\cl(\phi^\cl)=\true$, as required.

For the converse direction, let $\mathbf{v}$ be a~classical valuation s.t.\ $\mathbf{v}(\phi^\cl)=\mathbf{v}(\chi^\cl)=\true$ and $v(\psi^\cl)=\false$. We construct $\mathbf{v}^\four$ as in~\eqref{equ:vfour}. This gives us that $\mathbf{v}^\four(\chi)\in\{\true,\both\}$ but $v^\four(\psi)\notin\{\true,\both\}$. We check that $v^\four(\phi)\in\{\true,\both\}$. Since $\mathbf{v}(\phi^\cl)=\true$, we immediately obtain that $\mathbf{v}^\four(p_i)\in\{\true,\both\}$ and $\mathbf{v}^\four(p_{i'})\in\{\false,\both\}$. Additionally, as $\mathbf{v}(p^\circ_j)=\true$, we have that $\mathbf{v}(p^+_j)\neq v(p^-_j)$. I.e., $\mathbf{v}^\four(p_j)\in\{\true,\false\}$, whence $\mathbf{v}^\four(\circ p_j)=\true$, as required. A~similar argument can be used to show that $\mathbf{v}^\four(\bullet p_{j'})=\true$. The result now follows.
\end{proof}
\subsection{Proof of Theorem~\ref{theorem:BDcircCPLabduction}\label{ssec:BDcircCPLabductionproof}}
Let $\mathbb{P}=\langle\Gamma,\chi,\Hmsf\rangle$ be a~$\BD$ abduction problem and $\Xi=\Prop[\Gamma\cup\{\chi\}]$. Then $\phi$~is its (theory-minimal) proper $\Lcirc$-solution iff $\phi^{\sim}$ is a~(theory-minimal) proper solution of $\mathbb{P}^\cl_\circ\!=\!\left\langle\Gamma^\cl\!\cup\!\left\{\phi^{\leftrightarrow}_\Xi\right\},\phi^{\leftrightarrow}_{\Xi}\rightarrow\chi^\cl,\Hmsf^\cl_\circ\right\rangle$ with $\Hmsf^\cl_\circ=\{p^+\mid p\in\Hmsf\}\cup\{p^-\mid\neg p\in\Hmsf\}\cup\{p^\circ\mid\circ p\in\Hmsf\}\cup\{{\sim}p^\circ\mid\bullet p\in\Hmsf\}$.
\begin{proof}
Let $\phi$ be an atomic $\Lcirc$-term. We are going to show that the following statements hold for $\phi^\sim$ (cf.~Definitions~\ref{def:CPLabductiveproblem} and~\ref{def:atomiccirctermtoCPL}).
\begin{enumerate}
\item $\Gamma^\cl,\phi^\leftrightarrow_\Xi,\phi^\sim\not\models_\CPL\bot$ iff $\Gamma,\phi\not\models_\BD\bot$.
\item $\Gamma^\cl,\phi^\leftrightarrow_\Xi,\phi^\sim\models_\CPL\phi^\leftrightarrow_\Xi\rightarrow\chi^\cl$ iff $\Gamma,\phi\models_\BD\chi$.
\item 
$\phi^\sim\not\models_\CPL\phi^\leftrightarrow_\Xi\rightarrow\chi^\cl$ iff $\phi\not\models_\BD\chi$.
\item There is no proper solution $\phi'$ of $\mathbb{P}^\cl$ s.t.\ $\Gamma^\cl,\phi^\leftrightarrow_\Xi,\phi^\sim\models_\CPL\phi'$ but $\Gamma^\cl,\phi^\leftrightarrow_\Xi,\phi' \not \models_\CPL\phi^\sim$ iff there is no proper solution $\tau$ of $\mathbb{P}$ s.t.\ $\Gamma,\phi\models_\BD\tau$ but $\Gamma,\tau\not\models_\BD\phi$. 
\end{enumerate}

Point 1 follows immediately from Lemma~\ref{lemma:BDcirctoCPLatomicterms} because $\psi$ is \emph{classically unsatisfiable} (and similarly, $\BD$-unsatisfiable) iff it entails a~fresh variable $p^+$. But in this case, $\Gamma^\cl,\phi^\leftrightarrow_\Xi,\phi^\sim\not\models_\CPL p^+$ iff $\Gamma,\phi\not\models_\BD p$ for a~fresh $p$. This is equivalent to $\Gamma,\phi\not\models_\BD\bot$, as required.

To check Point 2, we observe that it is equivalent (via the deduction theorem and contraction) to the following entailment: $\Gamma^\cl,\phi^\leftrightarrow_\Xi,\phi^\sim\models_\CPL\chi^\cl$. Again, by Lem\-ma~\ref{lemma:BDcirctoCPLatomicterms}, this is equivalent to $\Gamma,\phi\models_\BD\chi$. Similarly, Point 3 is equivalent to $\phi^\sim,\phi^\leftrightarrow_\Xi\not\models_\CPL\chi^\cl$. By Lemma~\ref{lemma:BDcirctoCPLatomicterms}, this is equivalent to $\phi\not\models_\BD\chi$, as required.

Points 1--3 imply that $\phi\!\in\!\PropSol(\mathbb{P})$ iff $\phi^\sim\!\in\!\PropSol(\mathbb{P}^\cl_\circ)$. It remains to show that the same holds for \emph{theory-minimal solutions} too.

To show Point 4, assume for a contradiction that there is another proper solution $\phi'$ of $\mathbb{P}^\cl_\circ$ s.t.\ $\Gamma^\cl,\phi^\leftrightarrow_\Xi,\phi^\sim\models_\CPL\phi'$. Note from the construction of $\Hmsf^\cl_\circ$ that all negative literals in $\phi'$ are ${\sim}p^\circ$. Thus, there is some atomic $\Lcirc$-term $\tau$ s.t.\ $\tau^{\sim}=\phi'$ and $\tau^\leftrightarrow_\Xi=\phi^\leftrightarrow_\Xi$. It suffices to show that the following equivalence holds.
\begin{align*}
\Gamma^\cl,\phi^\leftrightarrow_\Xi,\phi^\sim\models_\CPL\tau^\sim&\text{ iff }\Gamma,\phi\models_\BD\tau
\end{align*}
Let $\Gamma^\cl,\phi^\leftrightarrow_\Xi,\phi^\sim\not\models_\CPL\tau^\sim$ and let further, $v$ be a~falsifying valuation. We define $v^\four(p)$ as in~\eqref{equ:vfour}. Now, it is clear from the proof of Lemma~\ref{lemma:BDtriangletoCPL} (cf.~Section~\ref{ssec:BDtriangletoCPLproof}) that $v^\four[\Gamma]\in\{\true,\both\}$. It remains to check that $v^\four(\phi)\in\{\true,\both\}$ but $v(\tau)\notin\{\true,\both\}$. For this, it suffices to show that $v^\four(\circ p)=\true$ iff $v(p^\circ)=\true$. But observe that since $v(\phi^\leftrightarrow_\Xi)=\true$, we have that $v(p^\circ)=\true$ iff $v(p^+)\neq v(p^-)$, i.e., (by definition of $v^\four$) iff $v^\four(p)\in\{\true,\false\}$, which is equivalent to $v^\four(\circ p)=\true$, as required.

Conversely, let $\Gamma,\phi\not\models_\BD\tau$ and $v$ be a~falsifying valuation. We define $v^\cl$ as in~\eqref{equ:vcl}. Again, it is clear that $v^\cl[\Gamma^\cl]=\true$. To verify that $v^\cl(\phi^\leftrightarrow_\Xi)=\true$ and $v^\cl(\phi^\sim)=\true$, we check that $v^\cl({\sim}p^\circ\leftrightarrow(p^+\leftrightarrow p^-))=\true$ for every $p^\circ\in\Hmsf$. For this, observe that $v(\circ p)=\true$ iff $v(p)\in\{\true,\false\}$, i.e., iff $v^\cl(p^+)\neq v^\cl(p^-)$. Thus, indeed, $v^\cl({\sim}p^\circ\leftrightarrow(p^+\leftrightarrow p^-))=\true$, as required.
\end{proof}
\end{document}